\documentclass[11pt,a4paper]{article}

\usepackage{amsmath,amsfonts,amssymb,amsthm}
\usepackage{amsthm}
\usepackage{graphicx,color}
\usepackage{boxedminipage}
\usepackage[ruled,boxed,linesnumbered]{algorithm2e}
\usepackage{framed}
\usepackage{thmtools}
\usepackage{thm-restate}
\usepackage{xspace}
\usepackage{todonotes}
\usepackage{footnote}
\usepackage{mathtools}
\usepackage{mathrsfs}
\usepackage{vmargin}
\setmarginsrb{1in}{1in}{1in}{1in}{0pt}{0pt}{0pt}{6mm}
\usepackage{todonotes}
\usepackage{footnote}
 \usepackage[pdftex, plainpages = false, pdfpagelabels, 
                 bookmarks=false,
                 bookmarksopen = true,
                 bookmarksnumbered = true,
                 breaklinks = true,
                 linktocpage,
                 pagebackref,
                 colorlinks = true,  
                 linkcolor = blue,
                 urlcolor  = blue,
                 citecolor = red,
                 anchorcolor = green,
                 hyperindex = true,
                 hyperfigures
                 ]{hyperref} 
 \usepackage{xifthen}
 \usepackage{tabularx}
\usepackage{tikz} 
 \usetikzlibrary{calc}
 \usepackage{thm-autoref}
\usepackage[nameinlink]{cleveref}

\newtheorem{lemma}{Lemma}

\newtheorem{corollary}{Corollary}
\newtheorem{definition}{Definition}
\newtheorem{observation}{Observation}
\newtheorem{proposition}{Proposition}

\theoremstyle{definition}

\newcommand{\opt}{{\sf OPT}}

\DeclareMathOperator{\operatorClassP}{P}
\newcommand{\classP}{\ensuremath{\operatorClassP}}
\DeclareMathOperator{\operatorClassNP}{NP}
\newcommand{\classNP}{\ensuremath{\operatorClassNP}}

\DeclareMathOperator{\operatorClassFPT}{FPT\xspace}
\newcommand{\classFPT}{\ensuremath{\operatorClassFPT}\xspace}

\DeclareMathOperator{\operatorClassFPTAS}{FPT-AS\xspace}
\newcommand{\classFPTAS}{\ensuremath{\operatorClassFPTAS}\xspace}

\newlength{\RoundedBoxWidth}
\newsavebox{\GrayRoundedBox}
\newenvironment{GrayBox}[1]%
   {\setlength{\RoundedBoxWidth}{.93\textwidth}
    \def\boxheading{#1}
    \begin{lrbox}{\GrayRoundedBox}
       \begin{minipage}{\RoundedBoxWidth}}%
   {   \end{minipage}
    \end{lrbox}
    \begin{center}
    \begin{tikzpicture}%
       \node(Text)[draw=black!20,fill=white,rounded corners,%
             inner sep=2ex,text width=\RoundedBoxWidth]%
             {\usebox{\GrayRoundedBox}};
        \coordinate(x) at (current bounding box.north west);
        \node [draw=white,rectangle,inner sep=3pt,anchor=north west,fill=white] 
        at ($(x)+(6pt,.75em)$) {\boxheading};
    \end{tikzpicture}
    \end{center}}

\newenvironment{defproblemx}[2][]{\noindent\ignorespaces%
                                \FrameSep=6pt%
                                \parindent=0pt%
                \vspace*{-1.5em}
                \ifthenelse{\isempty{#1}}{%
                  \begin{GrayBox}{\textsc{#2}}%
                }{%
                  \begin{GrayBox}{\textsc{#2} parameterized by~{#1}}%
                }
                \begin{tabular*}{\textwidth}{@{\hspace{.1em}} >{\itshape} p{1.8cm} p{0.8\textwidth} @{}}%
            }{
                \end{tabular*}%
                \end{GrayBox}%
                \ignorespacesafterend
            }

\newcommand{\defproblema}[3]{
  \begin{defproblemx}{#1}
    Input:  & #2 \\
    Task: & #3
  \end{defproblemx}
}%

\newcommand{\Oh}{\mathcal{O}}

\newcommand{\lr}[1]{\left( #1\right)}

\newcommand{\cP}{\mathcal{P}}
\newcommand{\cS}{\mathcal{S}}
\newcommand{\cL}{\mathcal{L}}

\newcommand{\pname}{\textsc}
\newcommand{\ProblemFormat}[1]{\pname{#1}}
\newcommand{\ProblemIndex}[1]{\index{problem!\ProblemFormat{#1}}}
\newcommand{\ProblemName}[1]{\ProblemFormat{#1}\ProblemIndex{#1}{}\xspace}

\newcommand{\probPack}{\ProblemName{Disk Repacking}}
\newcommand{\probCPack}{\ProblemName{Disk Appending}}
\newcommand{\probMPack}{\ProblemName{Max Disk Repacking}}

\begin{document}

\title{(Re)packing Equal Disks into Rectangle\thanks{The research leading to these results has received funding from the Research Council of Norway via the project  BWCA (grant no. 314528), the European Research Council (ERC) via grant LOPPRE, reference 819416, and Israel Science Foundation (ISF) grant no. 1176/18.} \footnote{A preliminary version of this article \cite{FominG0Z22} appears in the proceedings of the 49th International Colloquium on Automata, Languages, and Programming, ICALP 2022.}
}

\author{
Fedor V. Fomin\thanks{
Department of Informatics, University of Bergen, Norway.}
\and
Petr A. Golovach\addtocounter{footnote}{-1}\footnotemark{}
\and
Tanmay Inamdar\addtocounter{footnote}{-1}\footnotemark{}
\and
Saket Saurabh\addtocounter{footnote}{-1}\footnotemark{} \thanks{The Institute of Mathematical Science, HBNI, Chennai, India}
\and
Meirav Zehavi~\thanks{Ben-Gurion University of the Negev, Beer-Sheva, Israel} 
}

\date{}

\maketitle

\begin{abstract}
The problem of packing of equal  disks (or circles) into a rectangle  is a fundamental geometric problem.  (By a packing here we mean  an   arrangement of disks in a rectangle without overlapping.)
We consider the following  algorithmic generalization  of the equal disk packing problem.
In this  problem, for a given packing of equal disks into a rectangle, the question is whether by changing positions of a small number of disks, we can allocate space for packing more disks. More formally, in the repacking problem,  for a given set of $n$ equal disks packed into a rectangle and integers $k$ and $h$, we ask whether it is possible by changing positions of at most  $h$ disks  to  pack $n+k$ disks.  Thus  the problem of packing equal  disks is the special case of our problem with  $n=h=0$.

While the computational complexity of  packing equal  disks into a rectangle remains open, we prove that the repacking problem is 
NP-hard already  for $h=0$. Our main algorithmic contribution is   an algorithm that solves the repacking problem in time $(h+k)^{\Oh(h+k)}\cdot |I|^{\Oh(1)}$, where $I$ is the input size. 
That is, the problem is fixed-parameter tractable parameterized by $k$ and $h$.



\end{abstract}


\section{Introduction}\label{sec:intro}
Packing of equal circles inside a rectangle or a square is one of the oldest packing problems. In addition to many common-life applications, like packing bottles or cans in a box \cite{Goldberg70},  packings of circles  have a variety of industrial applications, including  circular cutting problems, communication networks, facility location,  and dashboard layout. 
We refer to the survey of Castillo,  Kampas,   and Pintér 
\cite{castillo2008solving} for an interesting overview of industrial applications of circle packings. 
 
The mathematical study of packing equal circles can be traced to Kepler \cite{J.-Kepler:1611sf}. 
Packing of circles also poses exciting mathematical and algorithmic challenges. After the significant efforts spent on packing for several decades \cite{schaer1965densest,toth2013lagerungen,maranas1995new,nurmela1999more,locatelli2002packing,szabo2007new,nurmela1997packing,FeketeMS19}, optimal packings of  equal circles inside a square are known only for instances of up to tens of circles  \cite{croft2012unsolved,specht2015best}.
The computational complexity of packing of  equal circles (NP-hardness or membership in NP) remains elusive. For packing circles with different radii,  Demaine, Fekete, and  Lang 
claimed 
  NP-hardness~\cite{DemaineFJ10}. See also the work  of Abrahamsen, Miltzow, and Seiferth \cite{AbrahamsenMS20} for a generic framework for establishing $\exists \mathbb{R}$-completeness for packing problems.


Our paper establishes several results on computational and parameterized complexity of a natural generalization of packing equal circles inside a rectangle. A remark in the terminology is in order. In the literature on packing, both terms, circles and disks, could be found. While the term circle is much more popular than disk, we decided to use disks for the following reason: In our hardness proof, it is more convenient to operate with open disks.
%
 Thus all disks we consider  are open and  unit (that is of radius one).  
Let us remind, that a family of disks forms a \emph{packing} if they are pairwise nonintersecting.\footnote{In the literature, it is often required for geometric packings that a packing should be maximal. In particular, for disk packing, every disk should touch either the bounding rectangle or another disk. However, in our problem, the task is to add a specified number of new disks to a given family and this  makes the maximality condition in our case very artificial.} 
In our problem, we have a packing of disks in a rectangle, and the question is whether we can allocate some space for   more disks by relocating a small amount of disks. More precisely, we consider the following problem. See Figure~\ref{fig:introex} for an example.

\defproblema{\probPack}%
{A packing $\cP$ of $n$ unit disks inside a rectangle $R$ and two integers $h,k\geq 0$.}%
{Decide whether there is a packing $\cP^*$ of $n+k$ unit disks inside $R$ obtained from $\cP$ by adding $k$ new disks and relocating  at most $h$ disks of $\cP$ to new positions.}

\begin{figure} 
\begin{center}
\includegraphics[scale=1.5]{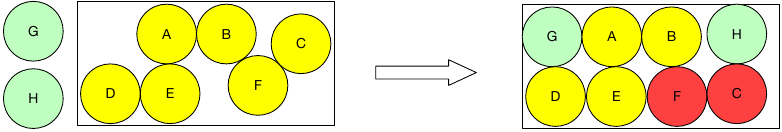}%
\caption{For a packing $\cP$ of  disks $A$--$F$, integers $h=2$, and $k=2$, the repacking  $\cP^*$ of $\cP$  is obtained by relocating disks $C$ and $F$, and by adding disks $G$ and $H$.   
 }\label{fig:introex}
\end{center}
\end{figure}

Thus when $n=0$, that is, initially there are no disks inside the rectangle, this is the classical problem of  packing equal circles inside a rectangle.

\subparagraph{Related Work on Geometric Packing.} Packing problems have received significant attention from the viewpoint of approximation algorithms. For the sake of illustration, let us mention a few examples. In $2$D Geometric Bin Packing, which is a variant of classical Bin Packing, the goal is to pack a given collection of rectangles into the minimum number of unit square bins. Typically, it is required that the rectangles be packed in an axis-parallel manner. There has been a long series of results on this problem, culminating in the currently known best approximation given by Bansal and Khan \cite{BansalK14}. A related problem is that of $2$D Strip Packing problem, where the task is to pack a given set of rectangles into an infinite strip of the given width, so as to minimize the height of packing. This problem has been studied from the context of approximation \cite{HarrenJPS14,JansenR19} as well as parameterized \cite{AshokKMS17} algorithms. Finally, we mention the Geometric Knapsack problem, which is also closely related to Geometric Bin Packing. In Geometric Knapsack, we are given a collection of rectangles, where each rectangle has an associated profit. The goal is to pack a subset of the given rectangles (without rotation) in an axis-aligned square knapsack, so as to maximize the total profit of the packed rectangles. Currently, the best approximation is given by Galvez et al.~\cite{GalvezGIHKW21}. A detailed survey of the literature on the results of these problems is beyond the scope of this work -- we direct an interested reader to the cited works and references therein and the survey paper of Christensen et al.~\cite{ChristensenKPT17}. However, we would like to highlight an important difficulty in \probPack---which is the focus of this work---as compared to the aforementioned geometric packing problems, namely, that packing disks in a rectangle requires the use of intricate geometric arguments as compared to packing rectilinear objects (such as rectangles) in a rectilinear container (such as a unit square, or an infinite strip). 

\subparagraph{Our Results.}
We show that \probPack is \classNP-hard even if the parameter $h=0$ -- we call this special case of problem \probCPack. 

\begin{restatable}{theorem}{themainhard}
\label{thm:compl-hard}
\probCPack is \classNP-hard when constrained to the instances $(R,\cP,k)$ where $R=[0,a]\times[0,b]$ for positive integers $a$ and $b$ and  
the centers of all disks in $\cP$ have rational coordinates. Furthermore, the problem remains \classNP-hard when it is only allowed to add new disks to $\cP$ with rational coordinates of their centers.  
\end{restatable}

From the positive side, we show that \probPack is \classFPT when parameterized by $k$ and $h$. 
As it is common in Computational Geometry, we assume the \emph{real RAM} computational model, that is, we are working with real numbers and assume that basic operations can be executed in unit time. 
We use $|I|$ to denote the input size. 

\begin{restatable}{theorem}{themainfpt}
\label{thm:repack-fpt}
The \probPack problem is \classFPT when parameterized by $k+h$.  Specifically, it is solvable in time $(h+k)^{\Oh(h+k)}\cdot |I|^{\Oh(1)}$. 
\end{restatable}

\Cref{thm:repack-fpt} also appears to be handy for  approximating the maximum number of disks that can be added to a packing.  
In the optimization variant of \probPack, called  \probMPack, we are given a packing $\cP$ of $n$ disks in a rectangle $R$ and an integer $h$, and the task is to maximize the number of new disks that can be added to the packing if we are allowed to relocate at most $h$ disks of $\cP$.  By combining \Cref{thm:repack-fpt} with the  approach of Hochbaum and Maass \cite{HochbaumM85}, we prove
  that the optimization variant of  \probPack  
  admits the parameterized analog of EPTAS for the parameterization by $h$.
  More precisely, we prove the following theorem. 

\begin{restatable}{theorem}{thefptas}
\label{cor:eptas}
For any $0 < \varepsilon < 1$, there exists an algorithm that, given an instance $(\cP,R,h)$ of \probMPack,  returns a packing $\cP^*$ into $R$ with at least $n + (1-\varepsilon) \cdot \opt_h$ disks in time $\big(\frac{h+1}{\varepsilon}\big)^{\Oh(h/\varepsilon+1/\varepsilon^2)} \cdot |I|^{\Oh(1)}$, where $\opt_h$ is the maximum number of disks that can be added to the input packing if we can relocate at most $h$ disks. 
\end{restatable}

\section{Preliminaries}\label{sec:prelim} 

\subparagraph{Disks and rectangles.} For two points $A$ and $B$ on the plane, we use $AB$ to denote the line segment with endpoints in $A$ and $B$. The \emph{distance} between 
$A=(x_1,y_1)$ and $B=(x_2,y_2)$  or the \emph{length} of $AB$, is $|AB|=\|A-B\|_2=\sqrt{(x_1-x_2)^2+(y_1-y_2)^2}$. 
The \emph{(open unit) disk} with a \emph{center} $C=(c_1,c_2)$ on the plane is the set of points $(x,y)$ satisfying the inequality $(x-c_1)^2+(y-c_2)^2<1$.  
Whenever we write ``disk'' we mean an open unit disk. 
Throughout the paper, we assume that considered input rectangles $R=[0,a]\times [0,b]$ for some $a,b>0$.

\subparagraph{Parameterized Complexity.} We refer to the book of Cygan et al.~\cite{CyganFKLMPPS15} for introduction to the area and undefined notions.  A \emph{parameterized problem} is a language $L\subseteq\Sigma^*\times\mathbb{N}$, where $\Sigma^*$ is a set of strings over a finite alphabet $\Sigma$. An input of a parameterized problem is a pair $(x,k)$, where $x\in\Sigma^*$ and $k\in \mathbb{N}$ is a \emph{parameter}. 
A parameterized problem is \emph{fixed-parameter tractable} (or \classFPT) if it can be solved in time $f(k)\cdot |x|^{\Oh(1)}$ for some computable function~$f$.  

\subparagraph{Systems of Polynomial Inequalities.} In our algorithms, we will need to find suitable locations for new disks that need to be added such that the locations are compatible with an existing packing. We will achieve this by solving systems of polynomial inequalities. We refer to the book of Basu, Pollack, and Roy~\cite{basu06} for basic tools. 
We use the following result.

\begin{proposition}[Theorem 13.13 in \cite{basu06}] \label{prop:polyequations}
	Let $R$ be a real closed field, and let $\mathcal{P} \subseteq R[X_1, \ldots, X_\ell]$ be a finite set of $s$ polynomials, each of degree at most $d$, and let \[(\exists X_1) (\exists X_2) \ldots (\exists X_\ell) F(X_1, X_2, \ldots, X_\ell)\]
	be a sentence, where $F(X_1, \ldots, X_\ell)$ be a quantifier-free boolean formula involving $\mathcal{P}$-atoms of type $P \odot 0$, where $\odot \in \{ =, \neq, > , < \}$, and $P$ is a polynomial in $\mathcal{P}$. Then, there exists an algorithm to decide the truth of the sentence with complexity $s^{\ell+1} d^{\Oh(\ell)}$ in $D$, where $D$ is the ring generated by the coefficients of the polynomials in $\mathcal{P}$. 
\end{proposition}

Furthermore, a point $(X_1^*,\ldots,X_\ell^*)$ satisfying $F(X_1, \ldots, X_\ell)$ can be computed in the same time by Algorithm~13.2 (sampling algorithm) of~\cite{basu06} (see Theorem~13.11 of \cite{basu06}). Note that because we are using the real RAM model in our algorithms, the complexity is stated with respect to the natural parameters.

 

\section{Hardness of \probCPack}\label{sec:repack-hard} 
In this section, we prove \Cref{thm:compl-hard} on the hardness of  \probCPack.   Recall, that \probCPack is the special case of  \probPack with $h=0$. We use the following auxiliary notation in this section.

 We use standard graph-theoretic terminology and refer to the textbook of Diestel~\cite{Diestel12} for missing notions. 
 We consider only finite undirected graphs.  For  a graph $G$,  $V(G)$ and $E(G)$ are used to denote its vertex and edge sets, respectively. 
 For a vertex $v\in V(G)$, $N_G(v)=\{u\in V(G)\mid uv\in E(G)\}$ denotes the \emph{neighborhood} of $v$, and $d_G(v)=|N_G(v)|$ is the \emph{degree} of $v$.
  A graph is \emph{cubic} if every vertex has degree three.   
A graph $G$ is \emph{planar} if it has a plane embedding, that is, it can be drawn on the plane without crossing edges.   
 A \emph{rectilinear} embedding is a planar embedding of $G$ such that vertices are mapped to points with integer coordinates and each edge is mapped into a broken line consisting of an alternate sequence of horizontal and vertical line segments. The switches between horizontal and vertical lines are called \emph{bends}. The \emph{area} of an embedding is minimum $(b_1-a_1)\times(b_2-a_2)$ such that all points of the embedding are in the rectangle $[a_1,b_1]\times [a_2,b_2]$.  

We say that a point $X$ is \emph{properly inside} of a polygon $P$ if it is inside $P$ but $X$ is not on the boundary; if we say that $X$ is inside $P$, we allow it to be on the boundary.    
A disk is \emph{(properly) iniside} of a polygon $ P$ if every point of the disk is (properly) inside of $P$.  

We restate the main theorem of the section. 

\themainhard*


\subparagraph{Proof of \Cref{thm:compl-hard}: Overview.}
We   reduce from  the \textsc{Independent Set} problem. Let us recall that in this problem, for a given a graph $G$ and a positive integer $k$, the task is to decide whether $G$ contains an independent set, that is a set of pairwise nonadjacent vertices,  of size at least $k$. 
It is well-known that \textsc{Independent Set} is \classNP-complete on cubic planar graphs~\cite{GareyJ79} (see also~\cite{Mohar01} for an explicit proof).

Before diving into the details, which are pretty technical, let us outline the main ideas of the reduction. 
Let $G$ be a graph and assume that $\ell_e$ are positive integers given for all $e\in E(G)$. Suppose that $G'$ is obtained from $G$ by subdividing each edge $e$ by $2\ell_e$ times (the edge subdivision operation for $e=uv$, deletes $e$ and creates a new vertex $w_e$ adjacent to both $u$ and $v$). Then it can be shown that $G$ has an independent set of size $k$ if and only if $G'$ has an independent set of size $k+\sum_{e\in E(G)}\ell_e$. We exploit this observation. 
Given a rectilinear embedding of a cubic planar graph $G$, for each vertex of $G$, we create a \emph{node} area formed by surrounding disks. We can place an additional disk in such an area and this encodes the inclusion of the corresponding vertex to an independent set.
Then we join the areas created for vertices by \emph{channels} corresponding to subdivided edges. Similarly to node areas, channels are formed by surrounding disks. 
Each channel contains even number of positions where new disks can be placed, and these positions are divided into ``odd'' and ``even'' in such a way that we can put disks in either all odd or all even positions but no disks could be placed in adjacent even and odd positions. Thus node areas and channels are used to encode a graph, and then we fill the space around them by \emph{filler} disks that prevent placing any new disk outside node areas and channels.   Then placing new disks corresponds to the choice of an independent set in a subdivided graph. Further in this section, we give a formal proof. To avoid unnecessary complications in the already technical proof, we allow algebraic number parameters in our reduction and then explain how we can get rid of these constraints. 

\subparagraph{Proof of \Cref{thm:compl-hard}: Constructing channels and node areas.}
Our construction of node areas and channels follows a rectilinear embedding of a planar graph and we use that fact that rectilinear embeddings can be constructed efficiently. In particular, the following theorem was shown by  Liu, Morgana, and Simeone~\cite{LiuMS98}.

\begin{proposition}[\cite{LiuMS98}]\label{thm:emb}
Every $n$-vertex planar graph of maximum degree at most $4$ admits a rectilinear embedding with at most $3$ bends for every edge with the area $\Oh(n^2)$. Furthermore, such an embedding can be constructed in $\Oh(n)$ time.
\end{proposition}

We use  \Cref{thm:emb} to construct node area and channels. Let $G$ be an $n$-vertex cubic graph. We assume that we are given a rectilinear embedding of $G$ with the properties guaranteed by \Cref{thm:emb}. We also assume without loss of generality that the length of every segment of a broken line representing an edge is at least three. 
This can be achieved by replacing every vertex or bend point $(x,y)$ of the embedding by the point $(3x,3y)$ and the corresponding adjustment of the segments in the broken lines.   
Notice that every segment in the embedding contains at least two integer points different from the endpoints of the segment.  
For each integer point of the rectangle containing the embedding, we construct a $2c\times 2c$ square tile, where $c$ is a sufficiently big odd positive integer (the choice of $c$ will be explained later), of one of the following four types:
(i) node tile containing a node area, (ii) horizontal/vertical  channel,  (iii) bend channel tiles to form channels, and (iv) filler tile to fill forbidden areas. Then we use these tiles to encode a graph as it is shown in \Cref{fig:emb} sticking the tiles together following the embedding.

\begin{figure}[ht]
\centering
\scalebox{0.7}{
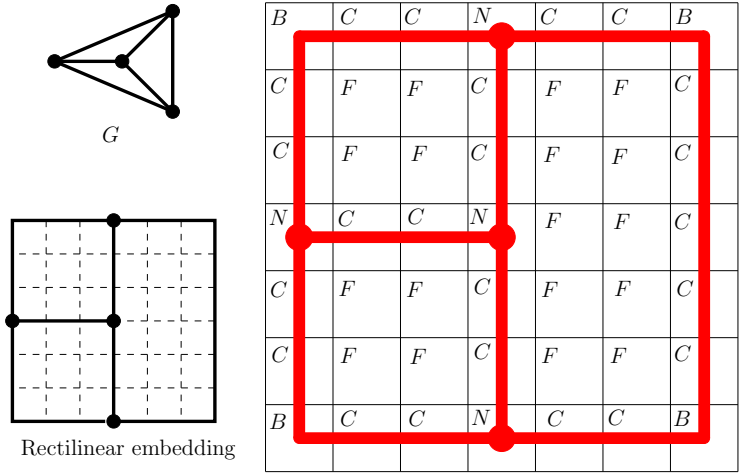
}
\caption{Encoding of $G$; the node areas and channels are shown in red, the node, channel, bend and filler tiles are labeled by $N$, $C$, $B$ and $F$, respectively.}\label{fig:emb}
\end{figure}

Now we describe these tiles. We start with the construction of the filler tile which is trivial---we simply fill a $2c\times 2c$ square by disks as it is shown in \Cref{fig:filler}. 

\begin{figure}[ht]
\centering
\scalebox{0.7}{
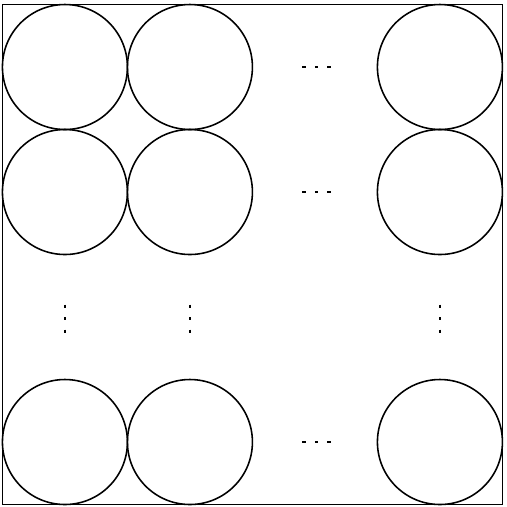
}
\caption{The filler tile.}\label{fig:filler}
\end{figure}

Next, we deal with channel tiles. The construction of these tiles is more complicated. In particular, we need three kinds of such tiles because we have to adjust parities and  join them together with other tiles. However, the basic idea is the same for all kinds. Consider four touching disks with centers $A$, $B$, $C$, and $D$ shown in \Cref{fig:channel-one} (a).  Note that 
$h=2+\sqrt{3}$, $\ell=|AC|=|BC|=2\sqrt{2+\sqrt{3}}$, and the angle $\alpha=\pi/12$. Then we can make the  straightforward observation that, given disks with centers in $A$, $B$ and $C$, every disk with its center in the triangle $ABC$ has its center in $D$. Then extending this, we can make the following observation about the configuration of disks shown in \Cref{fig:channel-one} (b). We call such a configuration of disks a \emph{basic channel} of size $r$.  When we say that a disk is placed or added, we mean that the disk should be disjoint with other disks. Also we say that a disk is \emph{inside} of a channel if its center is in $B_1A_1A_rB_r$.

 \begin{figure}[ht]
\centering
\scalebox{0.7}{
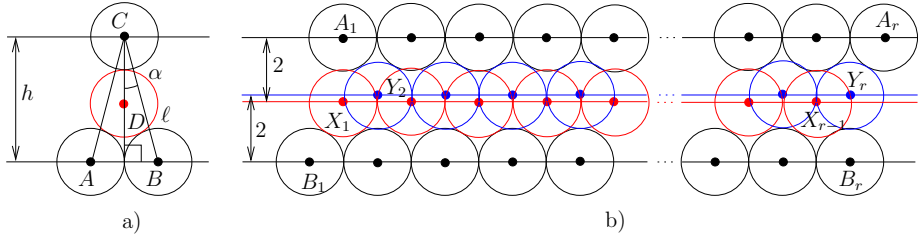
}
\caption{The basic channel of size $r$; the disks shown in red and blue are not parts of the channel -- they show places where new disks can be inserted.}\label{fig:channel-one}
\end{figure}

\begin{observation}\label{obs:channel-one}
Given disks with centers in $A_1,\ldots,A_r$ and $B_1,
\ldots,B_r$ as it is shown in \Cref{fig:channel-one}~(b), any disk placed properly inside the quadrilateral $A_1B_1B_rA_r$ has its center in one of the points 
$X_1,\ldots,X_{r-1}$ or $Y_2,\ldots,Y_r$. Furthermore, if a disk with its center in $X_i$ ($Y_i$, respectively) is placed in the  quadrilateral, then no other  disk can have its center in $Y_{i}$ or $Y_{i+1}$ ($X_{i-1}$ or $X_i$, respectively). 
\end{observation}

We use basic channels to construct channel, bend and node tiles.  In particular, we construct the \emph{straight}  channel tile from the basic channel of size $c$ by deleting the left bottom disk and filling the space outside the channel in the $2c\times 2c$ square by additional disks as it is shown in \Cref{fig:channel-tile} (a). 
The disks with the centers  in $A$ and $B$ are called \emph{poles}. They are identified with poles of other tiles to join them together. We refer to the basic channel inside the tile as the channel of the tile. 

 \begin{figure}[ht]
\centering
\scalebox{0.63}{
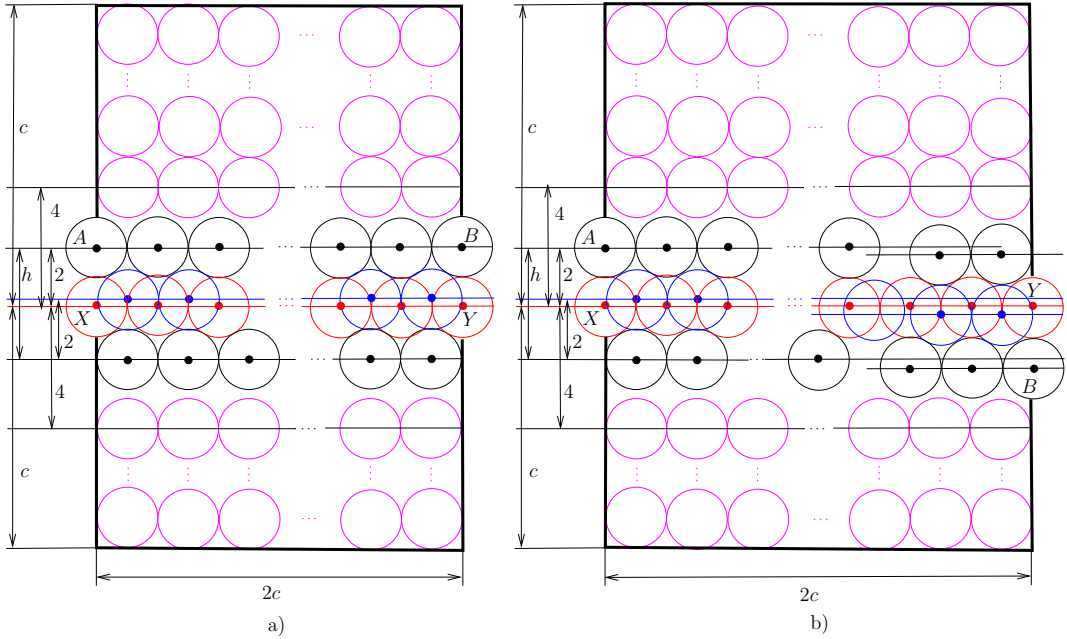
}
\caption{The straight channel tile (a) and the twisted channel tile (b). The disks shown in red and blue are not parts of the gadgets, the disks shown in magenta are used to fill space.}\label{fig:channel-tile}
\end{figure}

However, we need some further configurations of disks, because we have to adjust parities and distances in tiles, and also we have to joint tiles with each other. In particular, to join channel tiles with other tiles, we have to twist basic channels in some of them as it is shown in \Cref{fig:channel-twist} (a).  Then the following observation is straightforward.


 \begin{figure}[ht]
\centering
\scalebox{0.7}{
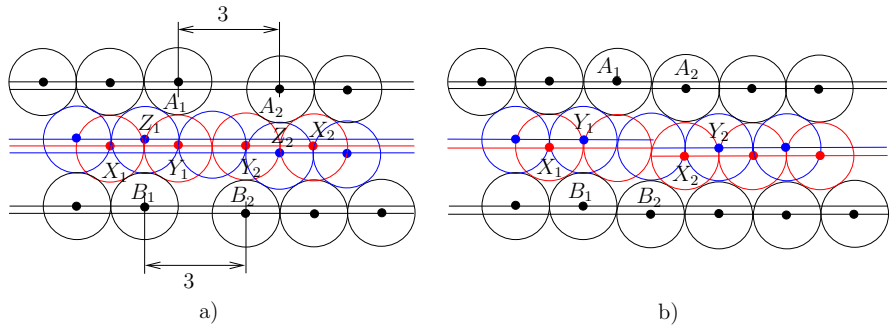
}
\caption{Twisting (a) and level adjustment (b).}\label{fig:channel-twist}
\end{figure}

\begin{observation}\label{obs:channel-twist}
Given disks with centers in $A_1,A_2$ and $B_1,B_2$ as it is shown in \Cref{fig:channel-twist} (a), the following holds:
\begin{itemize}
\item if two disks with their centers in $X_1$ and $X_2$ are placed as it is shown in the figure, then at most two disks with their centers inside $A_1B_1B_2A_2$ can be added, and if two disks are placed in $A_1B_1B_2A_2$ they have centers in $Y_1$ and $Y_2$, respectively,
\item if two disks with their centers in $Z_1$ and $Z_2$ are placed as it is shown in the figure, then at most one disk with its center $A_1B_1B_2A_2$ can be added.
\end{itemize}
\end{observation}

We construct the  \emph{twisted} channel tile (see \Cref{fig:channel-tile} (b)) similarly to the straight channel tile --  
the difference is that we insert one twist using \Cref{obs:channel-twist}. 
The crucial properties of straight and twisted channel tiles implied by   \autoref{obs:channel-one} and \autoref{obs:channel-twist} are given in the following lemma.
We say that a point is \emph{inside a tile} if it is inside of the $2c\times 2c$ square in the tile. 

\begin{lemma}\label{lem:channel}
At most $c+1$ new disks having their centers in the (straight, twisted) channel tile can be added  and it is possible to place $c+1$ disks. Moreover, the following holds:
\begin{itemize}
\item only disks inside channels can be added,
\item if exactly $c+1$ disks are placed, then two of them have their centers in $X$ and $Y$ (see \Cref{fig:channel-tile}),
\item it is possible to place $c$ disks that have no centers in $X$ and $Y$, but then they are completely inside the tile and it is impossible to place an additional 
disk having its center inside the tile. 
\end{itemize}
\end{lemma}

We need a modification of the straight tile to adjust parities. Notice that disks  placed inside a basic channel may be on different levels (see the red and blue disks in \Cref{fig:channel-one} (b) with their centers on the red and blue line, respectively). Hence, we need to adjust levels as it is shown in \Cref{fig:channel-twist} (b). Then we observe the following. 

\begin{observation}\label{obs:channel-level}
Suppose that we are given disks with centers in $A_1,A_2$ and $B_1,B_2$ as it is shown in \Cref{fig:channel-twist} (b). Then if 
there are two disks with their centers in $X_1$ and $X_2$ ($Y_1$ and $Y_2$, respectively), then at most one disjoint disk with its  center in $A_1B_1B_2A_2$ can be added.
\end{observation}

To fix parity, we also have to adjust distances. For this, we observe that we can insert gaps of length $s<\sqrt{4\sqrt{3}-3}-1$ between disks in basic channels as it is shown in \Cref{fig:channel-gap}~(a).      

 \begin{figure}[ht]
\centering
\scalebox{0.7}{
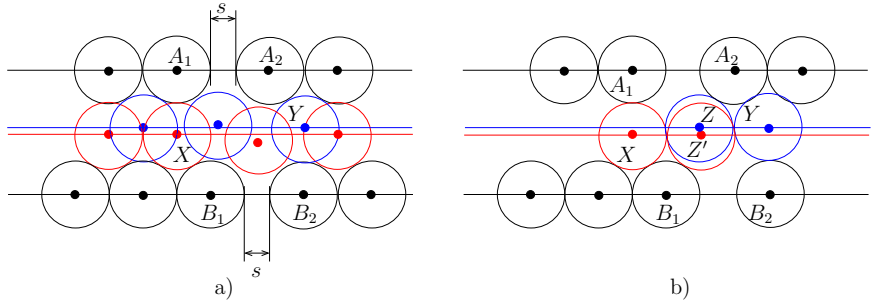
}
\caption{Gap insertion.}\label{fig:channel-gap}
\end{figure}

\begin{observation}\label{obs:channel-gap}
Given disks with centers in $A_1,A_2$ and $B_1,B_2$ as it is shown in \Cref{fig:channel-gap}~(a), at most one disk with its center inside the quadrilateral $A_1B_1B_2A_2$ can be added. Furthermore, if a disk has a center inside $A_1B_1B_2A_2$, then this disk intersects the disk with its center in $X$ or the disk with its center in $Y$. 
\end{observation}

\begin{proof}
The claim follows from the following geometrical observation illustrated in \Cref{fig:channel-gap}~(b). Suppose that the gap is exactly $\sqrt{4\sqrt{3}-3}-1$ and there are disks with centers in $A_1$, $A_2$, $B_1$, $B_2$, $X$ and $Y$.  Then any disk with its center inside $A_1B_1B_2A_2$ either has its center in $Z$ and touches the disks with centers in $X$, $A_2$ and $Y$ or has its center in $Z'$ and touches the disks with centers $X$, $B_1$ and $Y$; note that $Z$ and $Z'$ are uniquely defined by these touching conditions. Therefore, if $s<\sqrt{4\sqrt{3}-3}-1$, a new disk cannot be inserted in $A_1B_1B_2A_2$.
\end{proof}

Now we construct the \emph{parity adjustment} channel tile from the basic channel of size $c-1$ by introducing two gaps of size $1/2<\sqrt{4\sqrt{3}-3}-1$ and one level adjustment as  it is shown in \Cref{fig:channel-parity}. Similarly to \Cref{lem:channel}, we have the following properties by making use of  \Cref{obs:channel-one}, \Cref{obs:channel-level} and \Cref{obs:channel-gap}.

 \begin{figure}[ht]
\centering
\scalebox{0.7}{
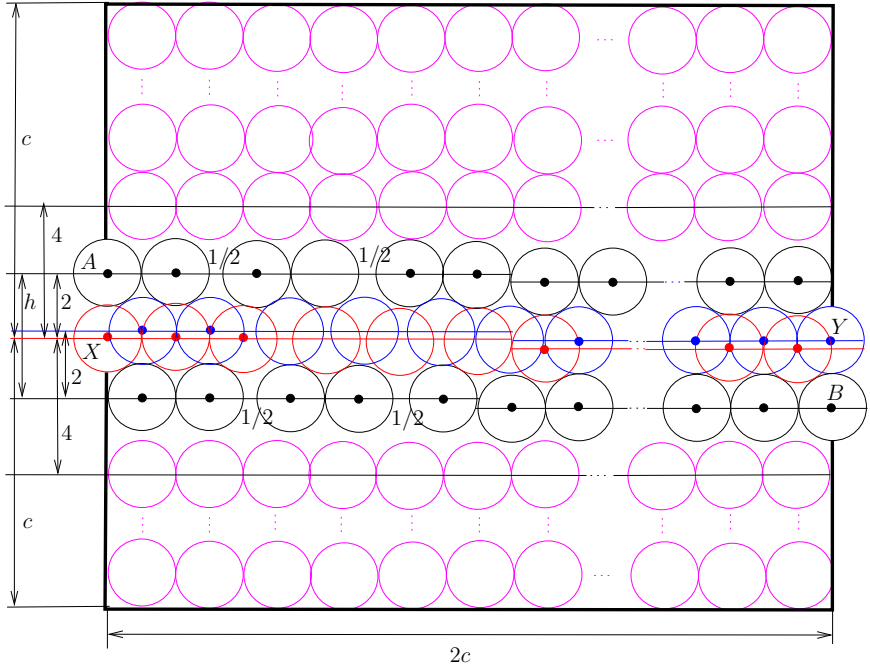
}
\caption{The parity adjustment channel tile. The disks shown in red and blue are not parts of the gadgets, the disks shown in magenta are used to fill space.}\label{fig:channel-parity}
\end{figure}

\begin{lemma}\label{lem:channel-parity}
At most $c$ new disks having their centers in the  parity adjustment channel tile can be added and 
 it is possible to place $c$ disks. Moreover,  only disks inside the channel of the tile can be added and 
  if $c$ disks are added then 
either one of them has its center in $X$ (see \Cref{fig:channel-parity}) and it is impossible to add the disk having its center in $Y$ or, symmetrically,  one disk has its center in $Y$ and the disk centered in $X$ cannot be inserted. 
\end{lemma}

 \begin{figure}[ht]
\centering
\scalebox{0.7}{
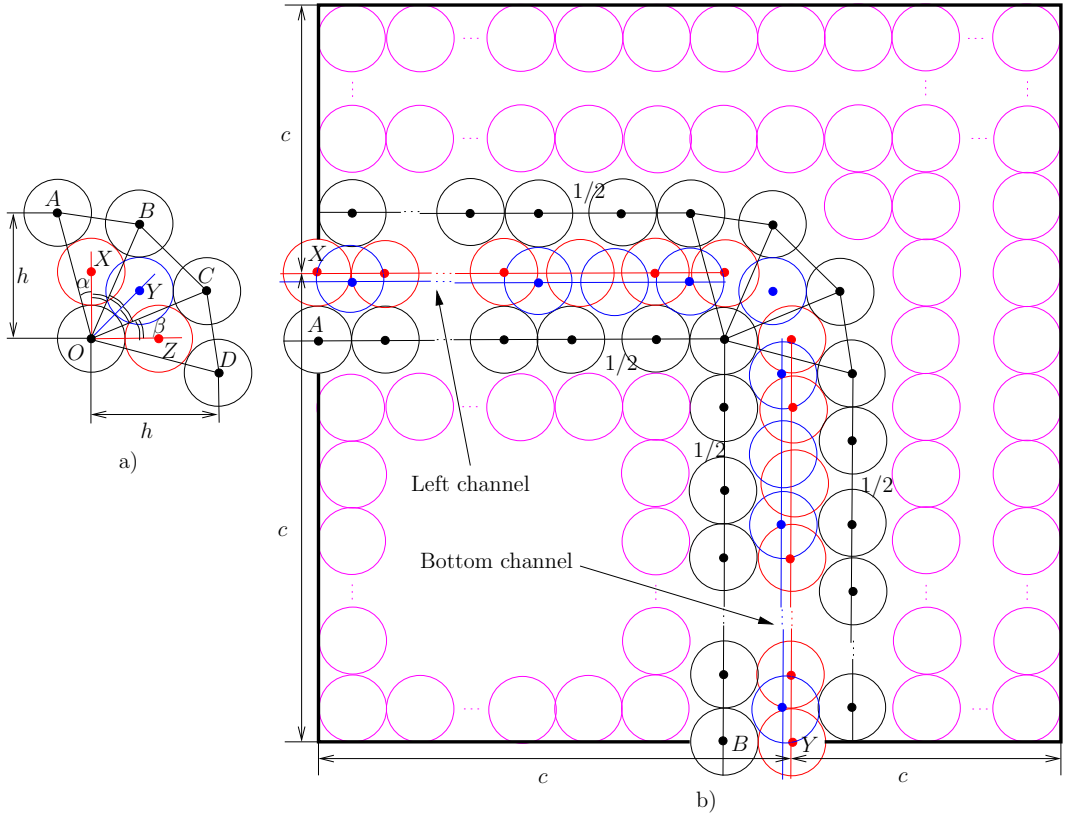
}
\caption{The bend tile. The disks shown in red and blue are not parts of the gadgets, the disks shown in magenta are used to fill space.}\label{fig:bend-tile}
\end{figure}

We use basic channels and apply gap insertions to construct the bend tile.  Additionally, we  observe that we can ``bend'' basic channels (see \Cref{fig:bend-tile} (a)).    
Consider five touching disks with centers $A$, $B$, $C$, $D$, and $O$ shown in \Cref{fig:bend-tile} (a);   $h=2+\sqrt{3}$, $|OA|=|OD|=\ell=2\sqrt{2+\sqrt{3}}$, $|OB|=|OC|=2\sqrt{2+\sqrt{2}}$, and the angles $\alpha=\pi/12$ and $\beta=\pi/8$. Then we can make the following observation.

\begin{observation}\label{obs:bend}
Given disks with their centers in $A$, $B$, $C$, $D$ and $O$, only disks with centers in $X$, $Y$ and $Z$ can have their centers in $ABCDO$. Moreover, if there is a disk with center in $X$ or $Z$, then the disk with its center in $Y$ cannot be added, and if there is the disc with its center in $Y$, then no disk having its center if $X$  or $Z$ can be added.
\end{observation}

\Cref{obs:bend} allows to construct the bend tile (see \Cref{fig:bend-tile} (b)). We use the configuration of disks from  \Cref{fig:bend-tile} (a) and attach two basic channels called \emph{left} and \emph{bottom} channels, respectively. To adjust distances, we insert two gaps of size $1/2$ into each channels. Then the remaining space if filled by disks. The disks with their centers in $A$ and $B$ are \emph{poles} of the tile. 
In the same way as with channel tiles, we have the following properties by making use of \Cref{obs:channel-one}, \Cref{obs:channel-gap} and \Cref{obs:bend}.   

\begin{lemma}\label{lem:bend}
At most $c-1$ new disks having their centers in the bend tile can be added and it is possible to place $c-1$ disks. Moreover, the following holds:
\begin{itemize}
\item disks can be placed only inside the channel,
\item if exactly $c-1$ disks are placed, then two of them have their centers in $X$ and $Y$ (see \Cref{fig:bend-tile} (a)),
\item it is possible to place $c-2$ disks that have no centers in $X$ and $Y$, but then they are completely inside the tile and it is impossible to place an additional 
disk having its center inside the tile. 
\end{itemize}
\end{lemma}
 
 \begin{figure}[ht]
\centering
\scalebox{0.7}{
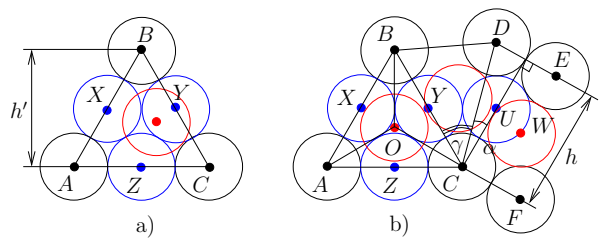
}
\caption{Node area.}\label{fig:node}
\end{figure}

The construction of the node tile is based on the following geometric observations.  Consider  an equilateral triangle $ABC$ with sides of length two as shown in \Cref{fig:node} (a), $h'=2\sqrt{3}$. Suppose that there are disks with centers in $A$, $B$ and $C$. Then it is possible to place at most three disks with centers in the triangle $ABC$,   and if exactly three disks are placed, then they have their centers in $X$, $Y$ and $Z$ and touch each other. Furthermore, if a disk having its center properly inside $ABC$ is placed, then no other disk with its center inside the triangle can be added. We exploit this property and add a basic channel as it is shown in \Cref{fig:node} (b). The point $O$ is the center of $ABC$, that is, $|OA|=|OB|=|OC|$. Recall that $h=2+\sqrt{3}$ and $\alpha=\pi/12$. We set $\gamma=\pi/3-\pi/12=\pi/4$. This gives us the configuration of disks with the following properties summarized in the next observation.

\begin{observation}\label{obs:node}
Given disks with centers in $A$, $B$, $C$, $D$, $E$ and $F$ as it is shown in \Cref{fig:node}~(b), the following is fulfilled:
\begin{itemize}
\item at most one disk with its center in $BCD$ can be added,
\item if there is a disk with its center either in $Y$ or $U$, then no other disk can have its center properly in $BCD$,
\item if there are disks with their centers in $O$ and $W$, then a disk with its center in $BCD$ can be added,
\item if there is a disk having its center properly inside $ABC$, then no other disk with its center inside $ABC$ can be added.
\end{itemize}
\end{observation}

We also using an easy observation that the basic channel construction allows to bend them by $\pi/6$ or $\pi/3$ as it is shown in \Cref{fig:bend-small}. Now we are ready to construct the node tile (see \Cref{fig:node-tile}).

 \begin{figure}[ht]
\centering
\scalebox{0.7}{
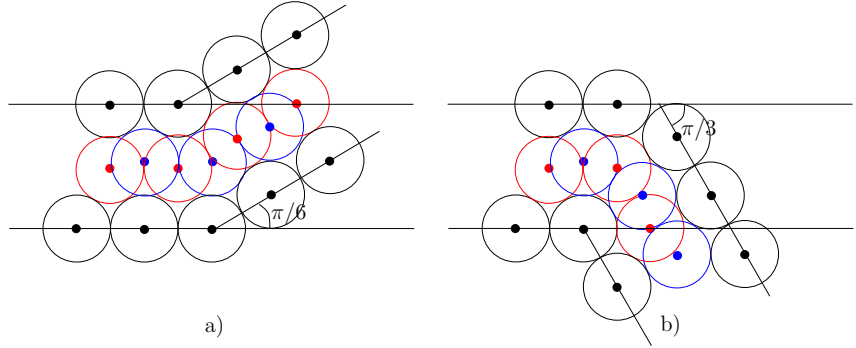
}
\caption{Bending of basic channels.}\label{fig:bend-small}
\end{figure}

 \begin{figure}[h!]
\centering
\scalebox{0.7}{
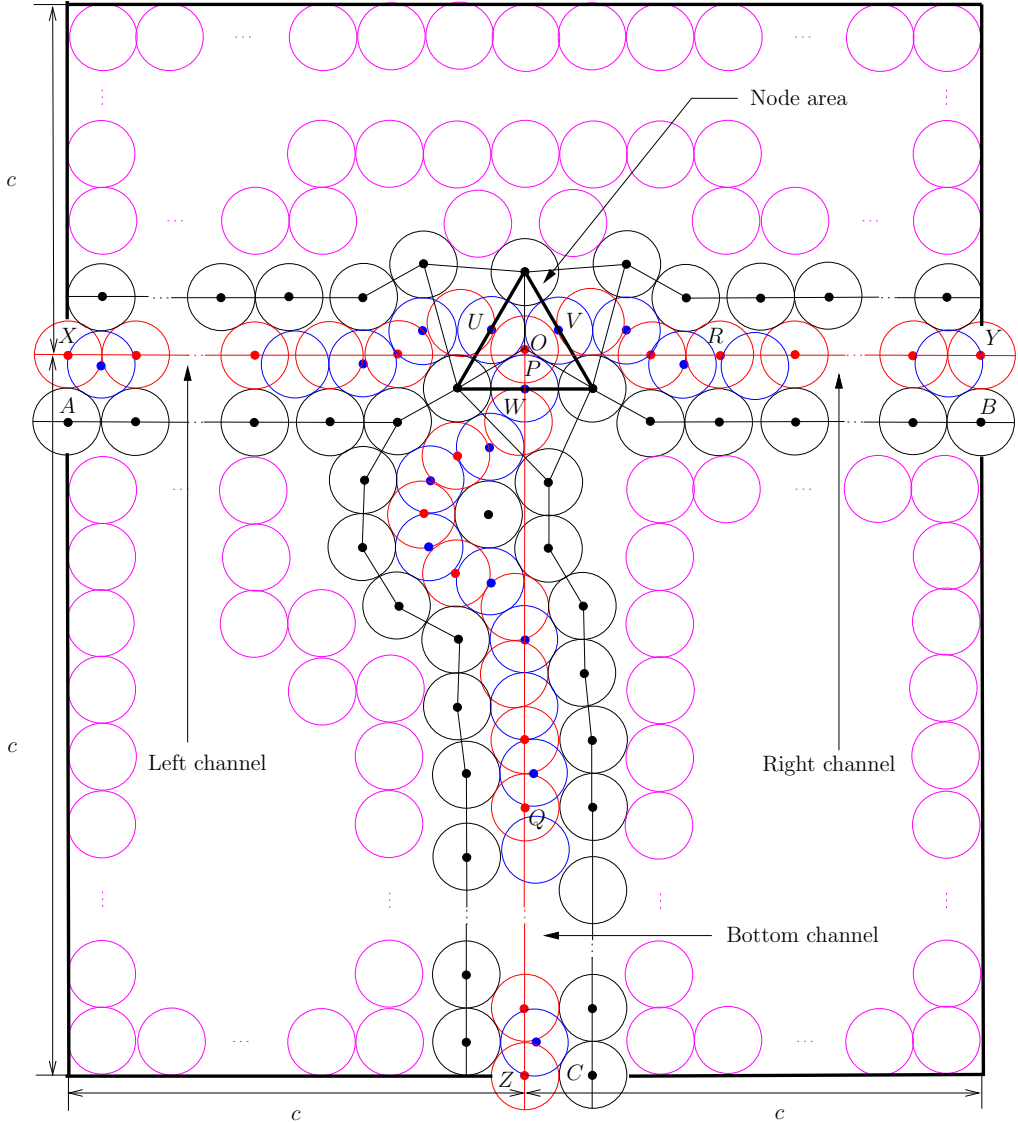
}
\caption{The node tile. The disks shown in red and blue are not parts of the gadgets, the disks shown in magenta are used to fill space. The point $P$ is the center of the tile.}\label{fig:node-tile}
\end{figure}

\begin{itemize}
\item We construct the node area formed by an equilateral triangle as it is shown in \Cref{fig:node}~(a).
\item We attach three basic channels to the node area as it is shown in \Cref{fig:node} (a); the channels are called \emph{left}, \emph{right} and \emph{bottom}, respectively, as it is shown in \Cref{fig:node-tile}.
\item To construct the left (right, respectively) channel, we use a basic channel with $\pi/6$ bend as is it is shown in \Cref{fig:bend-small}~(a). Notice that $|PR|=4+\sqrt{3}$, where $R$ is the point in the channel after the bend. To make the distances integer, we insert $2-\sqrt{3}$ gap in the basic channel (see \Cref{fig:channel-gap} (a)). Then we insert two gaps of length $1/2$ to adjust distances.
\item To construct the bottom channel, we bend the basic channel as it is it shown in \Cref{fig:bend-small}~(b) to adjust the direction. Then we make the level adjustment (see \Cref{fig:channel-twist} (b)). Note that $|PQ|=2\sqrt{3}+11$, where $Q$ is the point in the channel after the bending and level adjustment. We introduce $4-2\sqrt{3}$ gap and then we add $6$ gaps of length $1/2$ between parts of the basic channel to adjust distances and compensate that we can place more disks in bends that in a straight basic channel. 
\item The remaining space around the node area and channel is filled by disks as it is shown in \Cref{fig:node-tile}.   
 \end{itemize}
The disks with their centers in $A$, $B$ and $C$ are called \emph{poles} of the node tile. 

The properties of the node tile which follow from \Cref{obs:channel-level}, \Cref{obs:channel-gap},  and \Cref{obs:node} are summarized in the next lemma.

 \begin{lemma}\label{lem:node}
At most $(c-1)/2$ disks can be placed inside each channel and it is possible to place $(c-1)/2$ disks. Also at most one disks with its center inside the node area can be placed but it is possible to place the disk with its center in $O$.  
Furthermore, the following holds.
\begin{itemize}
\item Only disks inside the channels and the node area can be added. 
\item If there is a disk whose center is properly inside the node area, then at most $(c-1)/2$ disks can be placed in the left (right and bottom, respectively) channel. If exactly $(c-1)/2$ disks are placed, then one of the disks has its center in $X$ ($Y$ and $Z$, respectively). 
\item If there is a disk with its center in $U$ ($V$ and $W$, respectively), then at most $(c-1)/2$ disks (including the disk with its center in $U$ ($V$ and $W$, respectively)) can be placed in the left (right and bottom, respectively) channel, and if exactly $(c-1)/2$ disks are placed, then they are completely inside the tile and it is impossible to place an additional disk having its center inside the tile except disks that may be placed in other channels.
\end{itemize}
\end{lemma}

The construction of the node tile limits the choice of the constant $c$.

\begin{observation}\label{obs:c}
The (straight, twisted, parity adjustment) channel, bend, node tiles can be constructed for $c=47$.
\end{observation}

\begin{proof}
To construct the bottom channel in the node tile, we insert 7 gaps and a gap may be inserted between basic channels of size at least 2 (see \Cref{fig:channel-gap}~(a) and \Cref{obs:channel-gap}). Then taking into account the distance between the points $O$ and $Q$ in \Cref{fig:node-tile} and the number of gaps, we obtain that the bottom channel can be constructed for $c=47$. As for constructing the left and right channels in the node tile, we insert 3 gaps, we also have that they can be constructed for $c=47$.
By similar arguments, we also can construct the  (straight, twisted, parity adjustment) channel tile and the bend tile if $c=47$.
\end{proof}

\subparagraph{Proof of \Cref{thm:compl-hard}: The final step.}
Now we have all ingredients to finish the hardness proof for \probCPack.



\medskip
Recall that we prove \classNP-hardness by reducing from \textsc{Independent Set} on planar cubic graphs~\cite{GareyJ79,Mohar01}. Let $(G,k)$ be an instance of \textsc{Independent Set} where $G$ is an $n$-vertex planar cubic graph. Let us remind the initial steps.  By \Cref{thm:emb},  we can construct a rectilinear embedding of $G$ with area $\Oh(n^2)$ in linear time.  Further, we modify the embedding to ensure that the length of every segment of a broken line representing an edge in the embedding is at least three. 
As we already pointed, this can be done by replacing every vertex or bend point $(x,y)$ of the embedding by the point $(3x,3y)$ and the corresponding adjustment of the segments in the broken lines. After this modification we still have an embedding with $\Oh(n^2)$ area.   We assume that $R=[0,a]\times[0,b]$ for $a,b\in\mathbb{N}$ is the minimum area rectangle containing the embedding (note that $a,b>0$ because $G$ is cubic and cannot be embedded on the line). 

We define  $a'=2ca$ and $b'=2cb$, where $c=47$, and set $R'=[0,a']\times [0,b']$ defining the rectangle in the output instance of  \probCPack.  Then we put tiles into $R'$ as follows.
\begin{itemize}
\item For every $(x,y)\in R$ such that the point $(x,y)$ is not a point of the embedding of $G$, put a copy of the filler tile whose bottom left corner in $(2cx,2cy)$. 
\item For every $(x,y)\in R$ such that $(x,y)$ is  a vertex of $G$ in the embedding, put a copy of the node tile with the bottom left corner in $(2cx,2cy)$. We rotate the node tile in such a way that the directions of the left, right and bottom channels coincide with the directions of line segments of the embedding  with the endpoints in $(x,y)$; note that because the distance between any two vertices in the embedding is at least three, the poles of distinct node tiles do not interfere with each other. 
\item For every $(x,y)\in R$ such that $(x,y)$ is a bend node in the embedding of an edge, put a copy of the  bend tile with the bottom left corner in $(2cx,2cy)$. We rotate the tile in such a way that the directions of the left and bottom channels coincide with the directions of line segments of the embedding  with the endpoints in $(x,y)$. Again we note that because the distance between a bend point and another bend point or a vertex in the embedding is at least three, there is no intersections between the poles of constructed tiles. 
\item For every edge $e\in E(G)$, let $P_e$ be the set of internal non-bending integer points of the embedding of $e$. 

We select an arbitrary point $(x,y)\in P_e$ and insert a copy $T$ of the parity adjustment channel tile  with the bottom left corner in $(2cx,2cy)$. We rotate $T$ in such a way that the direction of its channel  coincides with the direction of line segments of the embedding  containing  $(x,y)$. Notice that because the length of every segment of a broken line representing an edge in the embedding is at least three, $T$ may be adjacent to at most one already placed tile $T'$ whose pole intersects $T$. If such a  pole of $T'$ has the same center as the corresponding pole of $T$, we unify these disks. Otherwise, if the poles have distinct centers, we reflect $T$ to ensure that the poles have the same centers and unify them.  

For every other point $(x,y)\in P_e$, we insert a tile $T$ which is a copy of either straight or twisted channel tile. We rotate $T$ to have the same direction of the channel as the direction of the segment of the line containing $(x,y)$ in the embedding.  Observe that $T$ can have ether one or two adjacent already placed tiles whose poles intersect $T$. 
If $T$ is not adjacent to any such a tile, we select $T$ be a copy of the straight tile.
If $T$ is adjacent to one such a tile $T'$, then we select $T$ be a copy of the straight channel tile. Then we ether identify the interfering poles of $T$ and $T'$ if they have the same centers or reflect $T$ and identify the poles afterwards. Suppose that $T$ is adjacent to two tiles $T'$ and $T''$ with interfering poles. If the poles are on the same side with respect to the channel, then we choose $T$ be a copy of the straight channel tile. Then we ether identify the interfering poles of $T$ and $T'$ if they have the same centers or reflect $T$ and identify the poles afterwards. Otherwise, we select $T$ be a copy of the twisted channel tile and reflect $T$ if necessary to identify the poles.

The construction  of the tiles for $(x,y)\in P_e$ is shown in \Cref{fig:edge}~(a). 
\end{itemize}

 \begin{figure}[h!]
\centering
\scalebox{0.6}{
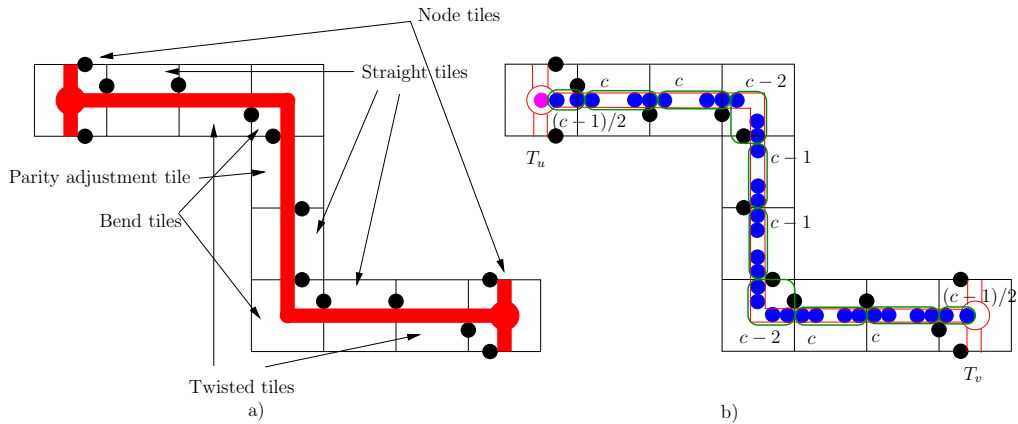
}
\caption{The construction of tiles for an edge and placement of tiles. The node areas and channels are shown in red and the poles are shown by black bullets. The placement of disk in tiles associated with the edge is shown in blue and the disk that may be placed in the center of $T_u$ is shown in magenta.}\label{fig:edge}
\end{figure}

We define $\cP$ be the set of all disks in the tiles (taking into account identifications of poles). By the construction, $\cP$ is a packing of disks inside $R'$.  

Clearly, we have $n$ node tiles. Denote by $n_b$ the number of bend tiles, by $n_p=|E(G)|$ the number of parity adjustment channel tiles, and by $n_c$ the number of straight and twisted channel tiles. We set $k'=k+3n(c-1)/2+(c-2)n_b+(c-1)n_p+cn_c$. 

We clam that $G$ has an independent set of size of size at least $k$ if and only if $(R',\cP,k')$ is a yes-instance of \probCPack.

For every edge $e$ of $G$, denote by $n_b^e$ the number of bend tiles and by $n_c^e$ the number of straight and twisted tiles in the set of tiles corresponding to the embedding of $e$.

For the forward direction, assume that $G$ has an independent set $S$ of size $k$. For every vertex $v\in S$, we consider the node tile $T_v$ corresponding to $v$ and place a disk having its center in the center of the tile (point $O$ in \Cref{fig:node-tile}). 
Consider an edge $e=uv$ of $G$.
Because $S$ is an independent set $u\notin S$ or $v\notin S$. Assume without loss of generality that $v\notin S$ and it may happen that $u\in S$. Then we can insert $(c-1)/2$ disks in the channels of $T_v$ and $T_u$ corresponding to $e$ by \Cref{lem:node}, $c-2$ disks per each bend tile by \Cref{lem:bend}, $c-1$ disks in the unique parity adjustment channel tile by \Cref{lem:channel-parity}, and $c$ disks per each straight or twisted channel tile by \Cref{lem:channel} as it is shown in \Cref{fig:edge}~(b) (we associate a pole disk shared by tiles with the first tile containing it along $e$ if moving from $u$ to $v$). Thus, we placed $2(c-1)/2+(c-2)n_b^e+(c-1)+cn_c^e$ disks. 
 Summarizing over all edges and taking into account the disks corresponding to the vertices of $S$, we obtain that we placed $k'=k+3n(c-1)/2+(c-2)n_b+(c-1)n_p+cn_c$  disks.  

For the opposite direction, assume that at least $k'$ disk can be placed in $R'$ to complement the packing $\cP'$. By \Cref{lem:channel,lem:node}, new disks can be only placed inside channels and node areas of the tiles.  
Let $\cS$ be a packing of $k'$ disks in $R'$ disjoint with the disks of $\cP$ such that the number of disks in $\cS$ whose centers are properly inside of  the node areas of node tiles is minimum. 

Consider an edge $e=uv$ of $G$. By  \Cref{lem:channel,lem:node}, at most $2(c-1)/2+(c-2)n_b^e+(c-1)+cn_c^e$ disks can be placed in  
the channels of $T_v$ and $T_u$ corresponding to $e$, the bend tiles, the unique parity adjustment channel tile and all straight or twisted channel tiles (see  \Cref{fig:edge}~(b) for an illustration). Moreover, for each $w\in V(G)$, at most one disk of $\cS$ can have its center properly inside of the node area of $T_w$. 
Summarizing over all edges, we conclude that at least $k$ disks of $\cS$ have their centers in the node areas of node tiles and for every $w\in V(G)$, at most one of these disks has its center inside of the node area of $T_w$. 

Suppose that there are two disks in $\cS$ such that their  
centers are inside of the node areas of $T_u$ and $T_v$.  Then by  \Cref{lem:channel,lem:node}, we conclude that at most 
$2(c-1)/2+(c-2)n_b^e+(c-1)+cn_c^e-1$ disks are placed in  the channels of 
 $T_v$ and $T_u$ corresponding to $e$ and other tiles associated with $e$. Then by \Cref{obs:node}, we can relocate the disk with its center in the node area of $T_v$ and move it to the channel of $T_v$ associated with $e$. Then we still would be able to place $2(c-1)/2+(c-2)n_b^e+(c-1)+cn_c^e$ disks by the same arguments as in the proof for the forward direction. This means that the relocation does not decrease the number of added disks. However, this contradicts our assumption about the choice of $\cS$ as we decrease 
 the number of disks with centers that are properly inside of  the node areas of node tiles. Hence, for every $e=uv$, there is no disk with its center in the node area of $T_u$ or $T_v$. 

Let $S\subseteq V(G)$ be the set of all vertices $w$ such  
that the node tile $T_w$ has a disk of $\cS$ with its center inside of the node area. We conclude that $S$ is an independent set $G$ of size at least $k$. This completes the proof of our claim.

The number of disks in each tile is at most $c^2$. This implies that each tile can be constructed in constant time.  The area of the rectilinear embedding of $G$ constructed  by the algorithm from \Cref{thm:emb} is $\Oh(n^2)$. Therefore, we construct $\Oh(n^2)$ tiles. Since the algorithm from \Cref{thm:emb}  is polynomial, we conclude that the instance $(R',\cP,k')$ of     \probCPack is constructed in polynomial time. This completes the proof \Cref{thm:compl-hard}.

\medskip

We proved \Cref{thm:compl-hard} assuming the real RAM model. In particular, to construct tiles we used disks with algebraic coordinates of their centers. However, we can observe that our construction is robust to allow rounding of coordinates. More precisely, we can choose a sufficiently small constant $\delta>0$ and use  rational parameters $\hat{h}$ and $\hat{h}'$ such that $2+\sqrt{3}=h<\hat{h}\leq h+\delta$ and $2\sqrt{3}=h'<\hat{h}'\leq h'+\delta$ 
in the construction of the basic channels (see \Cref{fig:channel-one}) and the node areas (see \Cref{fig:node}) instead of $h$ and $h'$, respectively. This observation  allows to obtain the following corollary.

\begin{corollary}\label{cor:compl-hard-rational}
	\probCPack is \classNP-hard when constrained to the instances $(R,\cP,k)$ where the centers all disks in $\cP$ have rational coordinates. Furthermore, the problem remains \classNP-hard when it is only allowed to add new disks to $\cP$ with rational coordinates of their centers.  
\end{corollary}

In \probPack, the area of the packing of disks is bounded by a rectangle $R$.  However, we can consider different types of boundaries. In particular, the boundary can be defined by disks in a given packing. Then we can obtain the following corollary.

\begin{corollary}\label{cor:compl-hard-diks}
	Given a rectangle $R=[0,2a]\times [0,2b]$ for integers $a,b>0$, a packing $\cP$ of disks with their centers inside $R$ such that (i) for every $i\in\{0,\ldots,a\}$, the disks with centers $(2i,0)$ and 
	$(2i,2b)$ are in $\cP$ and (ii) for   every $j\in\{0,\ldots,b\}$, the disks with centers $(0,2j)$ and 
	$(2a,2j)$ are in $\cP$, and an integer $k\geq 0$, it is \classNP-hard to decide whether $k$ disks with their centers in $R$ can be added to $\cP$ to form a packing. Moreover, the problem remains hard if $a=b$.
\end{corollary}


\newcommand{\hly}[1]{\sethlcolor{yellow}\hl{#1}\sethlcolor{yellow}}

\section{An FPT algorithm for  \probPack}\label{sec:repack-fpt} 
In this section, we prove  
that \probPack is \classFPT when parameterized by $k+h$.
\themainfpt*

 \subparagraph{Proof of \Cref{thm:repack-fpt}: Overview.}

On a high-level, the idea behind the algorithm is as follows. We first perform a greedy procedure to ensure that all ``free'' areas to place 
disks can be intersected by a set $\cal H$ of at most $k$ 
disks. Afterwards, we make use of a coloring function of $\cal P$ with the objective to color all disks in $\cal P$ that are repacked by a solution (if one exists) blue, and all disks in $\cal P$ that ``closely surround'' them by red. We need to ensure that, while relying on the initial greedy procedure, it would suffice to correctly color only $\Oh(h+k)$ disks. Indeed, this gives rise to the usage of a universal set, which is a ``small'' family of coloring functions ensured to contain, if there exists a solution, at least one coloring function that correctly colors all $\Oh(h+k)$ disks we care about.

Considering some coloring function (which 
expected to be  
 ``compatible'' with some solution), we identify ``slots'' and, more generally, ``containers'' in its coloring pattern. In simple words, a slot is just a 
disk in $R$ that does not intersect any red disk (from $\cal P$), and a container is a maximally connected region consisting of slots. We are able to prove that, if the coloring is compatible with some solution, then, for any container, either all or none of the disks in ${\cal P}$ that are contained in the container are repacked. This gives rise to a reduction from the problem of finding a solution compatible with a given coloring to the {\sc Knapsack} problem (more precisely, an extended version of it), where each container corresponds to an item whose weight is the number of disks in ${\cal P}$ that it contains, and whose value is the number of 
disks that can be packed within it. 

To execute the reduction described above, we need to be able to compute the value of each container. For this purpose, we first prove that a container can be ``described'' by only $\Oh(h+k)$ many disks from ${\cal P}\cup{\cal H}$; more precisely, we show that each container is the union of 
disks contained in $R$ that intersect at least one out of $\Oh(h+k)$ disks in ${\cal P}\cup{\cal H}$, from which we subtract the union of some other $\Oh(h+k)$ disks from ${\cal P}$. Having this at hand, to compute the value of a container, we first ``guess'', for each 
disk packed by  a (hypothetical) optimal packing of 
disks in the container, a disk from ${\cal P}\cup{\cal H}$ contained in the container (making use of its description) with whom it intersects. After that, we seek the corresponding optimal packing by making use of a system of polynomial equations (inequalities) of degree $2$, $\Oh(h+k)$ variables, and $\Oh((h+k)^2)$ equations.

\subparagraph{Proof of \Cref{thm:repack-fpt}: Free areas.}
To execute the plan above, we start with the task of handling the ``free'' areas. For this, we have the following definition and immediate observation.

\begin{definition}[{\bf Holes and Hole Cover}]
Let $({\cal P}, R,h,k)$ be an instance of \probPack. The {\em set of holes}, denoted by $\mathsf{Holes}$, is the set of all 
disks contained in $R$ that are disjoint from all disks in $\cal P$. A set ${\cal H}$ of 
disks contained in $R$ such that the set of holes of $({\cal P}\cup{\cal H}, R,h,k)$ is empty is called a {\em hole cover}.
\end{definition}

\begin{figure}
	\begin{minipage}{0.46\textwidth}
		\centering
		\vspace{-0.8cm}
		\includegraphics[scale=0.6]{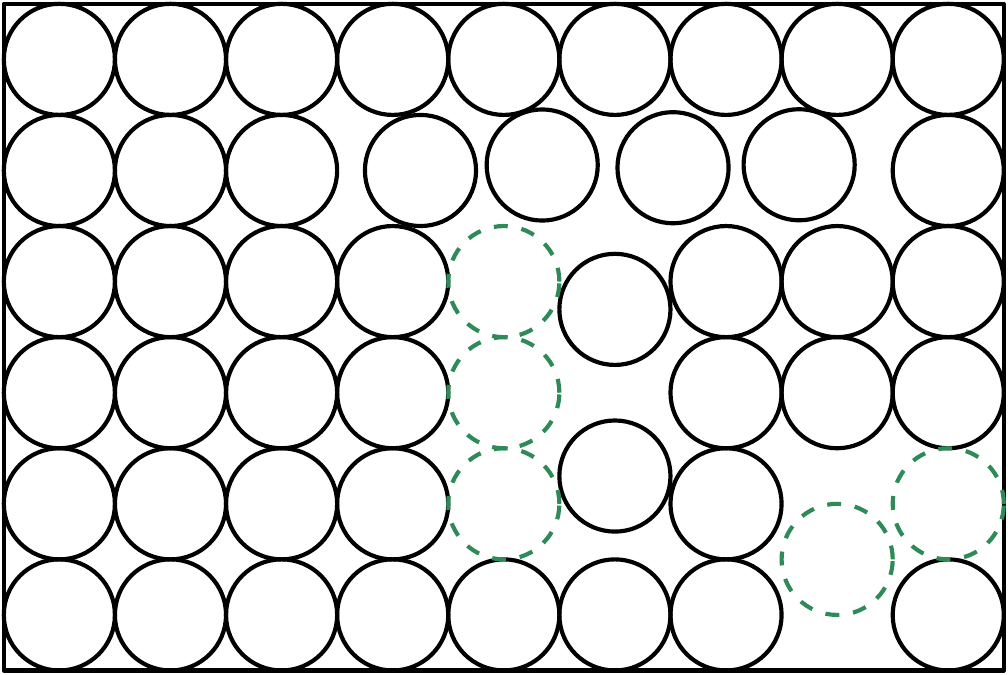}
		\caption{An instance $(\mathcal{P}, R, h = 2, k = 7)$ of \probPack. The disks in $\mathcal{P}$ are colored black. The disks in some hole cover $\mathcal{H}$ are colored green (using dashed lines).} \label{fig:repacking1}
	\end{minipage}\qquad
	\begin{minipage}{0.46\textwidth}
		\centering
		\includegraphics[scale=0.6]{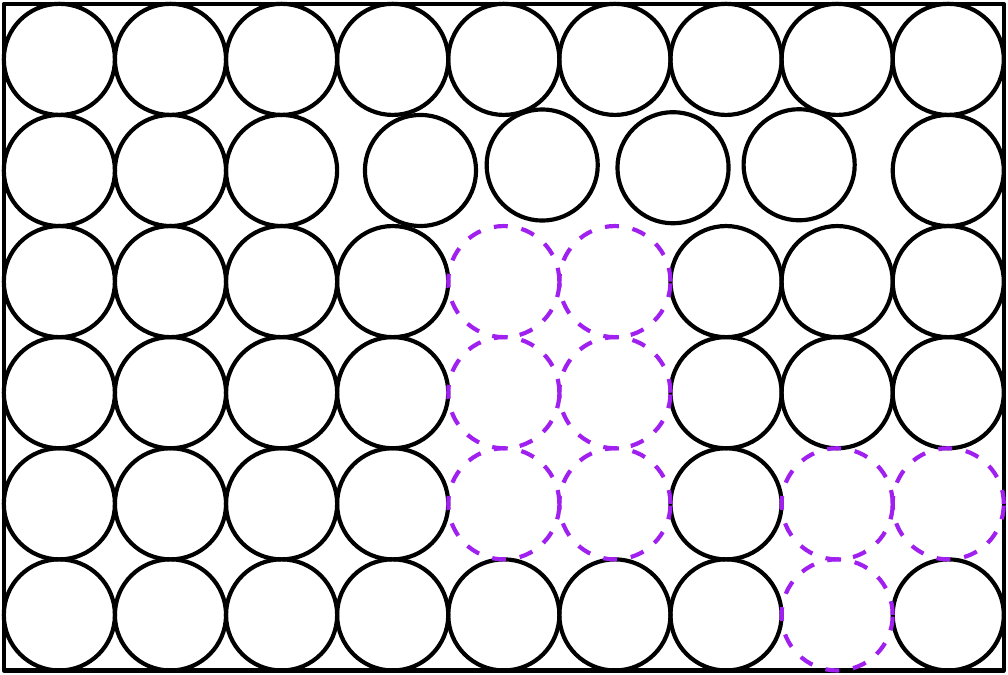}
		\caption{A solution $\mathcal{P}^*$ for the instance on the left. The disks in $\mathcal{P}^* \setminus \mathcal{P}$ are drawn in purple (using dashed lines).The set of $(\mathcal{H}, \mathcal {P}^*)$-critical disks is the set of green disks from the figure on the left and the purple disks from the figure on the right.} \label{fig:repacking2}
	\end{minipage}
\end{figure}

\begin{figure}
	\begin{minipage}{0.46\textwidth}
		\centering
		\includegraphics[scale=0.6]{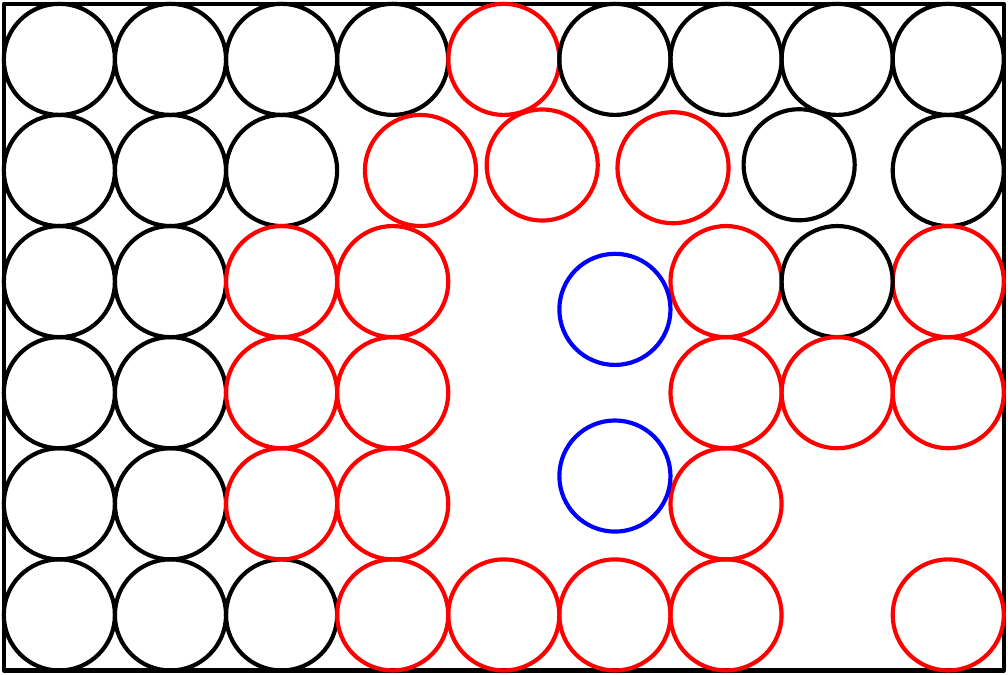}
		\caption{With respect to the instance and solution described in Figures \ref{fig:repacking1} and \ref{fig:repacking2}, the disks $(\cal{H}, \cal{P}^*)$-forced to be blue are colored blue, and the disks $(\mathcal{H},\mathcal{P}^*)$-forced to be red are colored red. Note that each of the disks colored black can be colored either blue or red by an $(\mathcal{H},\mathcal{P}^*)$-compatible coloring.} \label{fig:repacking3}
	\end{minipage}\qquad
	\begin{minipage}{0.46\textwidth}
		\centering
		\vspace{-0.8cm}
		\includegraphics[scale=0.6]{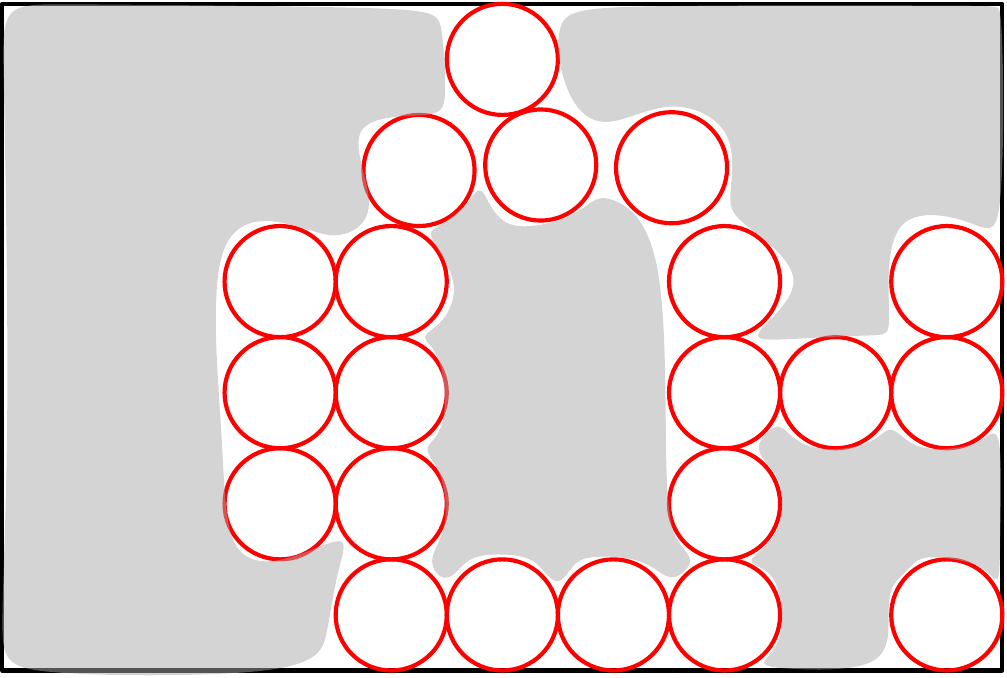}
		\caption{Consider an $(\mathcal{H}, \mathcal {P}^*)$-compatible coloring that colors blue all of the disks colored black in Figure \ref{fig:repacking3}. Then, we have four $c$-Containers, which roughly correspond to the areas colored by grey.} \label{fig:repacking4}
	\end{minipage}
\end{figure}

\begin{observation}\label{obs:cover}
Let $({\cal P}, R,h,k)$ be an instance of \probPack.  Let ${\cal H}$ be a hole cover. Then, every 
disk contained in $R$ intersects at least one disk in ${\cal P}\cup{\cal H}$.
\end{observation}

Next, we present a definition and a lemma that will allow us to assume that we have a hole cover of small size at hand. 

\begin{definition}[{\bf Dense instance}]
Let $({\cal P}, R,h,k)$ be an instance of \probPack. We say that the instance is {\em dense} if  it has a hole cover of size smaller than $k$.
\end{definition}

\begin{lemma}\label{lem:computeHoleCover}
There exists a polynomial-time algorithm that, given an instance $({\cal P}, R,h,k)$ of \probPack,  either  correctly determines that $({\cal P}, R,h,k)$ is a yes-instance or correctly determines that $({\cal P}, R,h,k)$ is dense and returns a hole cover of size smaller than $k$. 
\end{lemma}

\begin{proof}
We perform a simple greedy procedure. Initially, ${\cal H}=\emptyset$. Then, as long as there exists a 
disk $D$ contained in $R$ that is disjoint from all disks in ${\cal P}\cup{\cal H}$, we add such a disk $D$ to $\cal H$. The test for the existence of such a $D$ can be performed by using a system of polynomial equations of  degree $2$ with two variables denoting the $x$- and $y$-coordinates of the center of $D$. For each 
disk in ${\cal P}\cup{\cal H}$, we have an equation enforcing that the distance between its center and the center of $D$ is at least $2$, and additionally we have two linear equations to enforce that $D$ is contained in $R$. By \Cref{prop:polyequations}, testing whether this system has a solution (which corresponds to the sought disk $D$) can be done is polynomial time. Once the process terminates, the algorithm checks whether $|{\cal H}|\geq k$. If the answer is positive, then adding ${\cal H}$ (or, more precisely, any subset of size $k$ of it) to $\cal P$ is a solution, and so the algorithm answers yes, and otherwise the instance is dense and the algorithm returns $\cal H$ (which witnesses that).
\end{proof}

In the two following definitions, we identify the coloring functions that will be useful. 

\begin{definition}[{\bf $({\cal H},{\cal P}^*)$-Critical Disks}]
Let $({\cal P}, R,h,k)$ be a yes-instance of \probPack. Let ${\cal H}$ be a hole cover. Let ${\cal P}^*$ be a solution to $({\cal P}, R,h,k)$. The set of {\em $({\cal H},{\cal P}^*)$-critical disks}, denoted by $\mathsf{Crit}_{{\cal H},{\cal P}^*}$, is $({\cal P}^*\setminus {\cal P})\cup{\cal H}$.
\end{definition}

\begin{definition}[{\bf $({\cal H},{\cal P}^*)$-Compatible Colorings}]
Let $({\cal P}, R,h,k)$ be a yes-instance of \probPack. Let ${\cal H}$ be a hole cover.  Let ${\cal P}^*$ be a solution to $({\cal P}, R,h,k)$. Let $c: {\cal P}\rightarrow \{\mathsf{blue},\mathsf{red}\}$. We say that $c$ is {\em $({\cal H},{\cal P}^*)$-compatible} if:
\begin{enumerate}
\item For every $D\in{\cal P}\setminus{\cal P}^*$, we have that $c(D)=\mathsf{blue}$. We say that the disks in ${\cal P}\setminus{\cal P}^*$ are {\em $({\cal H},{\cal P}^*)$-forced to be blue}.
\item For every $D\in{\cal P}\cap{\cal P}^*$ whose center is at distance at most 4 from the center of some disk in $\mathsf{Crit}_{{\cal H},{\cal P}^*}$, we have that $c(D)=\mathsf{red}$. We say that the disks in ${\cal P}\cap{\cal P}^*$ whose center is at distance at most 4 from the center of some disk in $\mathsf{Crit}_{{\cal H},{\cal P}^*}$ are {\em $({\cal H},{\cal P}^*)$-forced to be red}.
\end{enumerate}
\end{definition}

We proceed to show that the number of disks in $\cal P$ that should be colored ``correctly'' is only   $\Oh(h+k)$. This is done using the following easy observation, in the following lemma.

\begin{observation}\label{obs:packingDisks}
The number of pairwise disjoint 
disks inside a circle of radius $r$ is at~most~$\pi r^2$.
\end{observation}


\begin{lemma}\label{lem:fewForced}
Let $({\cal P}, R,h,k)$ be a dense yes-instance of \probPack. Let ${\cal H}$ be a hole cover of size smaller than $k$.  Let ${\cal P}^*$ be a solution to $({\cal P}, R,h,k)$. Then, the number of disks $({\cal H},{\cal P}^*)$-forced to be either blue or red is altogether bounded by $\Oh(h+k)$.
\end{lemma}

\begin{proof}
Because ${\cal P}^*$ is a solution and $|{\cal H}|<k$, we have that
$|{\cal P}\setminus{\cal P}^*|\leq h$. So, at most $h$ disks are $({\cal H},{\cal P}^*)$-forced to be blue.
Further, $|\mathsf{Crit}_{{\cal H},{\cal P}^*}|=|({\cal P}^*\setminus {\cal P})\cup{\cal H}|<h+2k$. Observe that every 
disk in ${\cal P}\cap{\cal P}^*$ whose center is at distance at most 4 from the center of some disk in  $\mathsf{Crit}_{{\cal H},{\cal P}^*}$ is contained inside a circle of radius $5$ whose center is the center of some disk in  $\mathsf{Crit}_{{\cal H},{\cal P}^*}$. So, due to \Cref{obs:packingDisks} and since the disks in ${\cal P}\cap{\cal P}^*$ are pairwise disjoint, there exist at most $\pi\cdot 5^2\cdot (h+2k)=\Oh(h+k)$ disks in ${\cal P}\cap{\cal P}^*$ whose center is at distance at most 4 from the center of some disk in $\mathsf{Crit}_{{\cal H},{\cal P}^*}$. In particular, this means that at most $\Oh(h+k)$ disks are $({\cal H},{\cal P}^*)$-forced to be red.
This completes the proof.
\end{proof}

\subparagraph{Proof of \Cref{thm:repack-fpt}: Values of containers.}
Next, we present the definition of slots and containers, in which we will aim to (re)pack 
disks. The definition is followed by an observation and a lemma, which, in particular, state that if we try to repack at least one disk in a container, we can just repack all disks in that container.

\begin{definition}[{\bf  $c$-Slots and $c$-Containers}]
Let $({\cal P}, R,h,k)$ be an instance of \probPack. Let $c: {\cal P}\rightarrow \{\mathsf{blue},\mathsf{red}\}$. The set of {\em $c$-slots}, denoted by $\mathsf{Slots}_c$, is the set of 
 disks contained in $R$ that are disjoint from all disks in $\cal P$ that are colored red by $c$. The set of {\em $c$-containers}, denoted by $\mathsf{Containers}_c$, is the set of maximally connected regions in the union of all disks in $\mathsf{Slots}_c$.
\end{definition}

\begin{observation}\label{obs:disjointContainers}
Let $({\cal P}, R,h,k)$ be an instance of \probPack. Let $c: {\cal P}\rightarrow \{\mathsf{blue},\mathsf{red}\}$. Then, the regions in $\mathsf{Containers}_c$ are pairwise disjoint.
\end{observation}

\begin{lemma}\label{lem:allOrNone}
Let $({\cal P}, R,h,k)$ be a yes-instance of \probPack. Let ${\cal H}$ be a hole cover.  Let ${\cal P}^*$ be a solution to $({\cal P}, R,h,k)$. Let $c: {\cal P}\rightarrow \{\mathsf{blue},\mathsf{red}\}$ be $({\cal H},{\cal P}^*)$-compatible. Then, for every region $X\in\mathsf{Containers}_c$, either all disks in ${\cal P}$ contained in $X$ belong to ${\cal P}\setminus{\cal P}^*$ or none of the disks in ${\cal P}\cup{\cal P}^*$ contained in $X$ belongs to $({\cal P}\setminus{\cal P}^*)\cup({\cal P}^*\setminus{\cal P})$.
\end{lemma}

\begin{proof}
Targeting a contradiction, suppose that there exists a disk $D$  contained in $X$ that belongs to $({\cal P}\setminus{\cal P}^*)\cup({\cal P}^*\setminus{\cal P})$ and a disk $D'$ contained in $X$ that belongs to ${\cal P}\cap{\cal P}^*$.  Let $\gamma$ be a curve connecting the centers of these disks that lies entirely inside $X$. By the definition of a $c$-container and due to \Cref{obs:cover}, every point of this curve 
contained in 
a 
disk that belongs to $X$ and intersects a disk in $\cal P$ colored blue by $c$ or a disk  in $\cal H$. So, there must exist a point on $\gamma$ that is the center of a 
disk $D^*$ that intersects both a disk $A$ contained in $X$ that belongs to  $({\cal P}\setminus{\cal P}^*)\cup{\cal H}$ and a disk $A'$ contained in $X$ that belongs to ${\cal P}\cap{\cal P}^*$. From the definition of a $c$-container, $A'$ is colored blue by $c$.  Moreover, note that the center of $A'$ is at distance at most $4$ from the center of $A$, since each of the centers of $A$ and $A'$ is at distance at most 2 from the center of $D^*$. However, since $c$ is $({\cal H},{\cal P}^*)$-compatible, $A'$ is $({\cal H},{\cal P}^*)$-forced to be red and hence it is colored red by $c$. Since $c$ cannot color a disk both blue and red, we have reached a contradiction. This completes the proof.
\end{proof}

We proceed to define the weight and value of a $c$-container, which will be required for the reduction of our problem to {\sc Knapsack}.

\begin{definition}[{\bf Weight, Validity and Value of Containers}]\label{def:weightValue}
Let $({\cal P}, R,h,k)$ be an instance of \probPack. Let $c: {\cal P}\rightarrow \{\mathsf{blue},\mathsf{red}\}$. Let  $X\in\mathsf{Containers}_c$. The {\em weight} of $X$ is the number of disks in $\cal P$ that it 
contains.
We say that $X$ is {\em valid} if its weight is at most $h$. The {\em value} of $X$ is the maximum number of 
disks that can be packed inside $X$.
\end{definition}

The following is a corollary of \Cref{lem:allOrNone}.

\begin{corollary}\label{cor:allValid}
Let $({\cal P}, R,h,k)$ be a yes-instance of \probPack. Let ${\cal P}^*$ be a solution to $({\cal P}, R,h,k)$. Let $c: {\cal P}\rightarrow \{\mathsf{blue},\mathsf{red}\}$ be $(\cal{H},\cal{P}^*)$-compatible. Then, every disk in $({\cal P}\setminus{\cal P}^*)\cup({\cal P}^*\setminus{\cal P})$ is a $c$-slot, and it is contained in a valid $c$-container.
\end{corollary}

Now, we define a way in which we can ``easily'' describe a container, and then prove that this way can be encoded compactly.

\begin{definition}[{\bf Descriptions of Containers}]
Let $({\cal P}, R,h,k)$ be an instance of \probPack. Let ${\cal H}$ be a hole cover. Let $c: {\cal P}\rightarrow \{\mathsf{blue},\mathsf{red}\}$. An {\em $\cal H$-description} (or, for short, {\em description}) of a region $X\in\mathsf{Containers}_c$ is a pair $({\cal D}_1,{\cal D}_2)$ of a subset ${\cal D}_1\subseteq{\cal P}\cup{\cal H}$ and a minimal subset ${\cal D}_2\subseteq{\cal P}$ such that $X$ equals the set of all points in $R$ at distance less than $2$ from at least one disk in ${\cal D}_1$ and at least $2$ from all disks in ${\cal D}_2$.
\end{definition}

\begin{lemma}\label{lem:smallDescription}
Let $({\cal P}, R,h,k)$ be an instance of \probPack. Let ${\cal H}$ be a hole cover. Let $c: {\cal P}\rightarrow \{\mathsf{blue},\mathsf{red}\}$.  Let $X\in\mathsf{Containers}_c$. Then, $X$ has at least one description $({\cal D}_1,{\cal D}_2)$. Moreover, every description $({\cal D}_1,{\cal D}_2)$ of $X$ satisfies $|{\cal D}_1|+|{\cal D}_2|=\Oh(h'+k')$ where $h'$ is the weight of $X$, and $k'$ is the number of disks in $\cal H$ contained in $X$.
\end{lemma}

\begin{proof}
By \Cref{obs:cover}, every $c$-slot intersects at least one disk in $\{D\in{\cal P}: c(D)=\mathsf{blue}\}\cup{\cal H}$ and is disjoint from all disks in $\{D\in{\cal P}: c(D)=\mathsf{red}\}$. Further, every point in every disk in $\{D\in{\cal P}: c(D)=\mathsf{blue}\}\cup{\cal H}$ is contained in a $c$-slot. So, it is immediate that $X$ has a description $({\cal D}_1,{\cal D}_2)$, and that $|{\cal D}_1|=\Oh(h'+k')$. Due to \Cref{obs:packingDisks} and since the 
disks in ${\cal P}\cup{\cal H}$ are pairwise disjoint, any circle of radius $5$ whose center is a center of some disk in $\{D\in{\cal P}: c(D)=\mathsf{blue}\}\cup{\cal H}$ can contain inside at most $\pi\cdot 5^2$ disks from $\{D\in{\cal P}: c(D)=\mathsf{red}\}$. Due to the minimality of ${\cal D}_2$ (which is a subset of $\{D\in{\cal P}: c(D)=\mathsf{red}\}$), every disk in it must be contained inside a circle of radius $5$ whose center is a center of some disk in $\{D\in{\cal P}: c(D)=\mathsf{blue}\}\cup{\cal H}$. Hence, $|{\cal D}_2|\leq |{\cal D}_1|\cdot\pi\cdot 5^2=\Oh(h'+k')$.
\end{proof}

Next, we use a description in order to efficiently compute the value of a $c$-container.

\begin{lemma}\label{lem:computeValue}
There is
an $(h+k)^{\Oh(h+k)}\cdot |I|^{\Oh(1)}$-time algorithm that, given a dense instance $I=({\cal P}, R,h,k)$ of \probPack, a hole cover ${\cal H}$ of size smaller than $k$, $c: {\cal P}\rightarrow \{\mathsf{blue},\mathsf{red}\}$ and a valid region $X$ with a description $({\cal D}_1,{\cal D}_2)$, computes the value of $X$.
\end{lemma}

\begin{proof}
Given $I=({\cal P}, R,h,k), {\cal H}, c, X$ and $({\cal D}_1,{\cal D}_2)$, the algorithm works as follows. For $\ell=h+k,h+k-1,\ldots,1$, and for every vector $(D_1,D_2,\ldots,D_\ell)\in {\cal D}_1\times{\cal D}_1\times\cdots\times{\cal D}_1$, it tests whether there exist $\ell$ 
disks $S_1,S_2,\ldots,S_\ell$ such that, for every $i\in\{1,2,\ldots,\ell\}$, $S_i$ intersects $D_i$, is contained in $R$ and is disjoint from all disks in ${\cal D}_2$. The test is done by constructing a system of polynomial equations of degree $2$ with $2\ell$ variables and $\ell\cdot(|{\cal D}_2|+2)$ equations as follows. For every $i\in\{1,2,\ldots,\ell\}$, we have two variables, denoting the $x$- and $y$-coordinates of the center of $S_i$, one equation enforcing that the distance between the center of $S_i$ and the center of $D_i$ is smaller than $2$, $|{\cal D}_2|$ equations enforcing that the distance between the center of $S_i$ and the center of each of the disks in ${\cal D}_2$ is at least $2$, and two linear equations enforcing that $S_i$ is contained inside $R$. If the answer is positive, then the algorithm returns that the value of $X$ is $\ell$ and terminates; else, it proceeds to the next iteration. Observe that, when $\ell=1$, the algorithm necessarily terminates (since $X$ contains at least one $c$-slot).

The correctness of the algorithm is immediate from the definition of a description and the exhaustive search that it performs. As for the running time, first observe that, by \Cref{lem:smallDescription} and since $X$ is valid and $|{\cal H}|<k$, $|{\cal D}_1|+|{\cal D}_2|\leq\Oh(h+k)$. So, for a given $\ell$, we have $|{\cal D}_1|^{\Oh(\ell)}=(h+k)^{\Oh(h+k)}$ choices of vectors. Now, consider the iteration corresponding to some $\ell$ and some vector. Then, we solve a system of polynomial equations of degree $2$ with $\Oh(h+k)$ variables and $\Oh((h+k)^2)$ equations. By \Cref{prop:polyequations}, this can be done in time $(h+k)^{\Oh(h+k)}\cdot |I|^{\Oh(1)}$. Thus, the algorithm indeed runs in time $(h+k)^{\Oh(h+k)}\cdot |I|^{\Oh(1)}$.
\end{proof}

The following definition captures the set of all descriptions.

\begin{definition}[{\bf Blueprint}]
Let $({\cal P}, R,h,k)$ be an instance of \probPack. Let ${\cal H}$ be a hole cover. Let $c: {\cal P}\rightarrow \{\mathsf{blue},\mathsf{red}\}$. An {\em $({\cal H},c)$-blueprint} is 
a collection of pairs of sets $\mathsf{Blueprint}\subseteq 2^{{\cal P}\cup{\cal H}}\times 2^{\cal{P}}$, where the first elements of the pair are pairwise-disjoint subsets of $\cal{P}\cup\cal{H}$, such that each region in $\mathsf{Containers}_c$ has exactly one description in $\mathsf{Blueprint}$, and every pair in $\mathsf{Blueprint}$ is a description of a region in $\mathsf{Containers}_c$.
\end{definition}

Next, we show how to compute blueprints.

\begin{lemma}\label{lem:blueprint}
There exists a polynomial-time algorithm that, given an instance $({\cal P}, R,h,k)$ of \probPack, 
a hole cover ${\cal H}$, and $c: {\cal P}\rightarrow \{\mathsf{blue},\mathsf{red}\}$, outputs an $({\cal H},c)$-blueprint.
\end{lemma}

\begin{proof}
We will perform a simple greedy procedure to identify, for each disk in $\{D\in{\cal P}: c(D)=\mathsf{blue}\}\cup{\cal H}$, the description of the region that contains it. Observe that every $c$-container contains at least one disk in $\{D\in{\cal P}: c(D)=\mathsf{blue}\}\cup{\cal H}$ (due to \Cref{obs:cover} and the definition of a $c$-container). So, if for every disk $D\in\{D\in{\cal P}: c(D)=\mathsf{blue}\}\cup{\cal H}$ we will take exactly one description $({\cal D}_1,{\cal D}_2)$ among the descriptions we identified such that $D$ is contained in ${\cal D}_1$, we will obtain an $({\cal H},c)$-blueprint. 

To describe the greedy procedure, consider some $D\in \{D\in{\cal P}: c(D)=\mathsf{blue}\}\cup{\cal H}$. Let us first show how to attain ${\cal D}_1$. For this purpose, we initialize ${\cal D}_1=\{D\}$. Then, for every pair of disks $A\in {\cal D}_1$ and $B\in (\{D\in{\cal P}: c(D)=\mathsf{blue}\}\cup{\cal H})\setminus{\cal D}_1$, we test whether there exists a pair of 
disks $C$  and $C'$ that are contained in $R$, intersect each other, are disjoint from all disks in $\{D\in{\cal P}: c(D)=\mathsf{red}\}$, and such that $C$ intersects $A$ and $C'$ intersects $B$. The test for the existence of such a $C$ is performed by using a system of polynomial equations of  degree $2$ with four variables denoting the $x$- and $y$-coordinates of the centers of $C$ and $C'$. For each 
disk in $\{D\in{\cal P}: c(D)=\mathsf{red}\}$, we have two equations enforcing that the distances between its center and the centers of $C$ and $C'$ are each at least $2$. Additionally, we have three equations to enforce that the distance between the centers of $C$  and $C'$ is smaller than $2$, the distance between the centers of $C$  and $A$ is smaller than $2$, and the distance between the centers of $C'$  and $B$ is smaller than $2$, as well as four linear equations to enforce that $C$ and $C'$ are contained in $R$. By \Cref{prop:polyequations}, testing whether this system has a solution (which corresponds to the sought disks $C$ and $C'$) can be done is polynomial time. If the answer is positive, then we add $B$ to ${\cal D}_1$. In case at least one pair $(A,B)$ resulted in the addition of $B$  to ${\cal D}_1$, then we repeat the entire loop, iterating again over all pairs $(A,B)$ (where the domain from which they are taken is updated as a new disk was added to ${\cal D}_1$). Notice that we can perform at most $|{\cal P}|$ repetitions, and that each repetition results in at most $|{\cal P}\cup{\cal H}|^2$ many iterations, each taking polynomial time. Hence, the procedure, so far, runs in polynomial time. 

Now, let us show how to attain ${\cal D}_2$. For this purpose, we initialize ${\cal D}_2=\{D\in{\cal P}: c(D)=\mathsf{red}\}$. Now, for every $A\in\{D\in{\cal P}: c(D)=\mathsf{red}\}$, we test whether there exists a 
disk $C$ that is contained in $R$ and intersects both $A$ and at least one disk in ${\cal D}_1$, and is disjoint from all disks in ${\cal D}_2\setminus\{A\}$. The test can be performed by iterating over every disk $B\in{\cal D}_1$, and using a system of polynomial equations of  degree $2$ with two variables denoting the $x$- and $y$-coordinates of the center of $C$. For each 
disk in ${\cal D}_2\setminus\{A\}$, we have an equation enforcing that the distance between its center and the center of $C$ is at least $2$, and additionally we have two equations to enforce that the distance between the center of $C$ and each of the centers of $A$ and $B$ is smaller than $2$, as well as two linear equations to enforce that $C$ is contained in $R$. By \Cref{prop:polyequations}, testing whether this system has a solution (which corresponds to the sought disk $C$) can be done is polynomial time. If the answer is positive, then we remove $A$ from ${\cal D}_2$. Notice that this phase of the procedure also runs in polynomial time. Moreover, the correctness of the entire procedure directly follows from the definitions of a $c$-container and a description.
\end{proof}

We proceed to define the (extended version of the) {\sc Knapsack} problem and the instances of this problem that our reduction produces.

\begin{definition}[{\bf (Extended) Knapsack}]
In the {\sc (Extended) Knapsack} problem, we are given a collection of $n$ items $U$, where each item $u\in U$ has a weight $w(u)\in\mathbb{N}_0$ and a value $v(u)\in\mathbb{N}_0$, and an integer $W\in\mathbb{N}_0$. The objective is to find, for every $W'\in\{0,1,\ldots,W\}$, the maximum $V_{W'}\in\mathbb{N}_0$ for which there exists a subset of items $S\subseteq \{1,2,\ldots,n\}$ such that $\sum_{i\in S}w(u)\leq W'$ and $\sum_{i\in S}v(u)\geq V_{W'}$.
\end{definition}

\begin{definition}[{\bf $({\cal H},c)$-{\sc Knapsack} instance}]
Let $({\cal P}, R,h,k)$ be an instance of \probPack. Let ${\cal H}$ be a hole cover. Let $c: {\cal P}\rightarrow \{\mathsf{blue},\mathsf{red}\}$. The {\em $({\cal H},c)$-{\sc Knapsack} instance} is the instance $(U,w,v,W,V)$ of {\sc Knapsack} defined as follows: $U$ is the set of all valid regions in $\mathsf{Containers}_c$; for each $X\in U$, $w(X)$ and $v(X)$ are the weight and value of $X$ (see \Cref{def:weightValue}); $W=h$; $V=h+k$.
\end{definition}

\begin{proposition}[\cite{DBLP:books/daglib/0023376}]\label{prop:knapsack}
The {\sc (Extended) Knapsack} problem is solvable in time $\Oh(|U|\cdot W)$.
\end{proposition}

We now to prove the correspondence between our problem when we restrict the solution set to solutions compatible with a given coloring and the {\sc Knapsack} problem.

\begin{lemma}\label{lem:reductionPart}
Let $({\cal P}, R,h,k)$ be an instance of \probPack. Let ${\cal H}$ be a hole cover. Let $c: {\cal P}\rightarrow \{\mathsf{blue},\mathsf{red}\}$. Then, there exists a solution ${\cal P}^*$ to $({\cal P}, R,h,k)$ such that $c$ is compatible with ${\cal P}^*$ if and only if for the $({\cal H},c)$-{\sc Knapsack} instance $(U,w,v,W,V)$, there exists $W'\in\{0,1,\ldots,W\}$ such that $V_{W'}\geq W'+k$.
\end{lemma}

\begin{proof}
In one direction, suppose that there exists a solution ${\cal P}^*$ to $({\cal P}, R,h,k)$ such that $c$ is compatible with ${\cal P}^*$.  Let $X_1,X_2,\ldots,X_\ell$ be the $c$-containers that contain at least one disk from $({\cal P}\setminus{\cal P}^*)\cup({\cal P}^*\setminus{\cal P})$. By \Cref{obs:disjointContainers}, these $c$-containers are pairwise disjoint, by \Cref{lem:allOrNone} and since $c$ is compatible with ${\cal P}^*$, all disks in ${\cal P}$ contained in $X_1\cup X_2\cup\cdots\cup X_\ell$ belong to ${\cal P}\setminus{\cal P}^*$, and by \Cref{cor:allValid}  and since $c$ is compatible with ${\cal P}^*$, all disks in $({\cal P}\setminus{\cal P}^*)\cup({\cal P}^*\setminus{\cal P})$ are contained in $X_1\cup X_2\cup\cdots\cup X_\ell$ and all of these $c$-containers are valid. So, because ${\cal P}^*$ can repack $h$ disks from $\cal P$, the total weight of these $c$-containers must be some $W'\in \{0,1,\ldots,h\}=\{0,1,\ldots,W\}$, and since ${\cal P}^*$ also packs $k$ additional 
disks, the total value of these $c$-containers must be at least $W'+k$ (to accommodate all of the repacked and $k$ newly packed 
disks). Thus, $V_{W'}\geq W'+k$.

In the other direction, suppose that there exists $W'\in\{0,1,\ldots,W\}$ such that $V_{W'}\geq W'+k$. This means that there exist $c$-containers $X_1,X_2,\ldots,X_\ell$ whose total weight is $W'\in\{0,1,\ldots,h\}$ and whose total value is at least $W'+k$. However, because these $c$-containers are pairwise disjoint (by \Cref{obs:disjointContainers}), this means that we can construct a solution ${\cal P}^*$ such that $c$ is compatible with ${\cal P}^*$ by repacking all the disks in $\cal P$ that are contained in $X_1,X_2,\ldots,X_\ell$ (there are at most $h$ such disks) and, additionally, inserting $k$ new 
disks, within $X_1,X_2,\ldots,X_\ell$. This completes the proof.
\end{proof}

The following is a corollary of  \Cref{lem:computeValue,lem:blueprint}.

\begin{corollary}\label{cor:computeKnapsackInstance}
There exists an $(h+k)^{\Oh(h+k)}\cdot |I|^{\Oh(1)}$-time algorithm that, given a dense instance $I=({\cal P}, R,h,k)$ of \probPack, a hole cover ${\cal H}$ of size smaller than $k$ and  $c: {\cal P}\rightarrow \{\mathsf{blue},\mathsf{red}\}$, computes the $({\cal H},c)$-{\sc Knapsack} instance.
\end{corollary}

To compute coloring functions, we will use the following definition and proposition.

\begin{definition}[{\bf $(U, k)$-Universal Set}]
For a universe $U$ and $k\in\mathbb{N}$, a {\em $(U, k)$-universal set} is a collection $\cal C$ of functions $f: U\rightarrow\{\mathsf{blue},\mathsf{red}\}$ such that for every pair of disjoint sets $B,R\subseteq U$ whose union has size at most $k$, there exists $c\in{\cal C}$ that colors all integers in $B$ blue and all integers in $R$ red.
\end{definition}

\begin{proposition}[\cite{DBLP:conf/focs/NaorSS95}]\label{prop:universalSet}
There exists an algorithm that, given a universe $U$ of size $n$ and $k\in\mathbb{N}$, constructs a $(U, k)$-universal set of size $2^{k+\Oh(\log^2 k)}\log n$ in time $2^{k+\Oh(\log^2 k)}n\log n$.
\end{proposition}

Based on the definition of a universal set, we define the collection of {\sc Knapsack} instances relevant to our reduction.

\begin{definition}[{\bf $({\cal H},{\cal C})$-{\sc Knapsack} Collection}]\label{def:knapsackCollection}
Let $({\cal P}, R,h,k)$ be an instance of \probPack. Let ${\cal H}$ be a hole cover. Let ${\cal C}$ be a $({\cal P},q(h+k))$-universal set, where $q$ is the constant hidden in the $\Oh$-notation in \Cref{lem:fewForced}. Then, the {\em $({\cal H},{\cal C})$-{\sc Knapsack} collection} is the collection of {\sc Knapsack} instances that includes, for every $c\in{\cal C}$, the $({\cal H},c)$-{\sc Knapsack} instance.
\end{definition}

The following is a corollary of \Cref{cor:computeKnapsackInstance}.

\begin{corollary}\label{cor:computeKnapsackCollection}
There exists an $(h+k)^{\Oh(h+k)}\cdot |I|^{\Oh(1)}$-time algorithm that, given a dense instance $I=({\cal P}, R,h,k)$ of \probPack, a hole cover ${\cal H}$ of size smaller than $k$ and  a $({\cal P},q(h+k))$-universal set $\cal C$, computes the $({\cal H},{\cal C})$-{\sc Knapsack} collection.
\end{corollary}

Next, we prove the correspondence between our problem and the collection of {\sc Knapsack} instances we have just defined.

\begin{lemma}\label{lem:reductionFull}
Let $({\cal P}, R,h,k)$ be an instance of \probPack. Let ${\cal H}$ be a hole cover. Let ${\cal C}$ be a $({\cal P},q(h+k))$-universal set. Then, $({\cal P}, R,h,k)$ is a yes-instance of \probPack\ if and only if the $({\cal H},{\cal C})$-{\sc Knapsack} collection contains an instance $(U,w,v,W,V)$ for which there exists $W'\in\{0,1,\ldots,W\}$ such that $V_{W'}\geq W'+k$.
\end{lemma}

\begin{proof}
In one direction, suppose that $({\cal P}, R,h,k)$ is a yes-instance. By the definition of a $({\cal P},q(h+k))$-universal set and due to \Cref{lem:fewForced}, there exists $c\in{\cal C}$ that is compatible with ${\cal P}^*$. So, the $({\cal H},c)$-{\sc Knapsack} instance is contained in the $({\cal H},{\cal C})$-{\sc Knapsack} collection $(U,w,v,W,V)$, and by \Cref{lem:reductionPart}, for this instance there exists $W'\in\{0,1,\ldots,W\}$ such that $V_{W'}\geq W'+k$.

In the other direction, suppose that the $({\cal H},{\cal C})$-{\sc Knapsack} collection contains an instance $(U,w,v,W,V)$ for which there exists $W'\in\{0,1,\ldots,W\}$ such that $V_{W'}\geq W'+k$. This instance is a $({\cal H},c)$-{\sc Knapsack} instance for some $c\in{\cal C}$. So, by \Cref{lem:reductionPart}, $({\cal P}, R,h,k)$ is, in particular, a yes-instance of \probPack.
\end{proof}

%

\subparagraph{Proof of \Cref{thm:repack-fpt}: Putting it all together.}
We are now ready to make the final step of the proof  of \Cref{thm:repack-fpt}.

The algorithm works as follows. Given an instance $({\cal P}, R,h,k)$ of \probPack, it calls the algorithm in \Cref{lem:computeHoleCover} to either  correctly determine that $({\cal P}, R,h,k)$ is a yes-instance or correctly determine that $({\cal P}, R,h,k)$ is dense and obtain a hole cover $\cal H$ of size smaller than $k$.  In the first case, the algorithm is done. In the second case, the algorithm proceeds as follows. First, it calls the algorithm in \Cref{prop:universalSet} to obtain a $({\cal P},q(h+k))$-universal set $\cal C$. Then, it calls the algorithm in \Cref{cor:computeKnapsackCollection} to obtain the $({\cal H},{\cal C})$-{\sc Knapsack} collection. Afterwards, it uses the algorithm of \Cref{prop:knapsack} to determine whether the $({\cal H},{\cal C})$-{\sc Knapsack} collection contains an instance $(U,w,v,W,V)$ for which there exists $W'\in\{0,1,\ldots,W\}$ such that $V_{W'}\geq W'+k$.

The correctness of the algorithm follows from \Cref{lem:reductionFull}. The runtime bound of $(h+k)^{\Oh(h+k)}\cdot |I|^{\Oh(1)}$ follows from the runtimes bounds of the algorithms that the algorithm calls, stated in \Cref{lem:computeHoleCover}, \Cref{prop:universalSet}, \Cref{cor:computeKnapsackCollection}, and \Cref{prop:knapsack}.

This concludes the proof of \Cref{thm:repack-fpt}.
\section{Conclusion and open questions}\label{sec:conclusion} 
We have shown in \Cref{thm:compl-hard} that  \probPack problem is \classNP-hard even if $h=0$.  On the other hand, by \Cref{thm:repack-fpt},  \probPack is \classFPT when parameterized by $k$ and $h$. Both theorems naturally lead to the question about parameterization by $k$ only. The difficulty here is that even for  adding  one disk, one has to relocate many disks. Already for $k=1$, we do not know, whether the problem is in \classP \, or is \classNP-hard.

Another natural question stemming  from \Cref{thm:repack-fpt} is about kernelization of  \probPack. 
Does \probPack admit a polynomial kernel with parameters  $k$ and $h$? (We refer to books~\cite{CyganFKLMPPS15,FominLSZ19} for an introduction to kernelization).

Finally,  approximation  of   \probPack is an interesting research direction. In \Cref{cor:eptas} we demonstrated that our \classFPT algorithm can be used to construct an \classFPTAS with respect to $h$ for \probMPack. We leave open the question about polynomial approximation.  
Another open question concerns the approximability of the minimum number of relocations $h$ for a given $k$.
Already for $k=1$ finding a good approximation of $h$ is  a challenging problem.

\bibliographystyle{siam}
\bibliography{Disks}

\begin{thebibliography}{10}

\bibitem{AbrahamsenMS20}
{\sc M.~Abrahamsen, T.~Miltzow, and N.~Seiferth}, {\em Framework for
  er-completeness of two-dimensional packing problems}, in 61st {IEEE} Annual
  Symposium on Foundations of Computer Science (FOCS), {IEEE}, 2020,
  pp.~1014--1021.

\bibitem{AshokKMS17}
{\sc P.~Ashok, S.~Kolay, S.~M. Meesum, and S.~Saurabh}, {\em Parameterized
  complexity of strip packing and minimum volume packing}, Theor. Comput. Sci.,
  661 (2017), pp.~56--64.

\bibitem{BansalK14}
{\sc N.~Bansal and A.~Khan}, {\em Improved approximation algorithm for
  two-dimensional bin packing}, in Proceedings of the Twenty-Fifth Annual
  {ACM-SIAM} Symposium on Discrete Algorithms, {SODA} 2014, Portland, Oregon,
  USA, January 5-7, 2014, C.~Chekuri, ed., {SIAM}, 2014, pp.~13--25.

\bibitem{basu06}
{\sc S.~Basu, R.~Pollack, and M.-F. Roy}, {\em Algorithms in Real Algebraic
  Geometry}, Springer, Berlin, Heidelberg, 2009.

\bibitem{castillo2008solving}
{\sc I.~Castillo, F.~J. Kampas, and J.~D. Pint{\'e}r}, {\em Solving circle
  packing problems by global optimization: numerical results and industrial
  applications}, European Journal of Operational Research, 191 (2008),
  pp.~786--802.

\bibitem{ChristensenKPT17}
{\sc H.~I. Christensen, A.~Khan, S.~Pokutta, and P.~Tetali}, {\em Approximation
  and online algorithms for multidimensional bin packing: {A} survey}, Comput.
  Sci. Rev., 24 (2017), pp.~63--79.

\bibitem{DBLP:books/daglib/0023376}
{\sc T.~H. Cormen, C.~E. Leiserson, R.~L. Rivest, and C.~Stein}, {\em
  Introduction to Algorithms, 3rd Edition}, {MIT} Press, 2006.

\bibitem{croft2012unsolved}
{\sc H.~T. Croft, K.~Falconer, and R.~K. Guy}, {\em Unsolved problems in
  geometry: unsolved problems in intuitive mathematics}, vol.~2, Springer
  Science \& Business Media, 2012.

\bibitem{CyganFKLMPPS15}
{\sc M.~Cygan, F.~V. Fomin, L.~Kowalik, D.~Lokshtanov, D.~Marx, M.~Pilipczuk,
  M.~Pilipczuk, and S.~Saurabh}, {\em Parameterized Algorithms}, Springer,
  2015.

\bibitem{DemaineFJ10}
{\sc E.~D. Demaine, S.~P. Fekete, and R.~J. Lang}, {\em Circle packing for
  origami design is hard}, CoRR, abs/1008.1224 (2010).

\bibitem{Diestel12}
{\sc R.~Diestel}, {\em Graph Theory, 4th Edition}, vol.~173 of Graduate texts
  in mathematics, Springer, 2012.

\bibitem{FeketeMS19}
{\sc S.~P. Fekete, S.~Morr, and C.~Scheffer}, {\em Split packing: Algorithms
  for packing circles with optimal worst-case density}, Discret. Comput. Geom.,
  61 (2019), pp.~562--594.

\bibitem{FominG0Z22}
{\sc F.~V. Fomin, P.~A. Golovach, T.~Inamdar, and M.~Zehavi}, {\em (re)packing
  equal disks into rectangle}, in 49th International Colloquium on Automata,
  Languages, and Programming, {ICALP} 2022, July 4-8, 2022, Paris, France,
  M.~Bojanczyk, E.~Merelli, and D.~P. Woodruff, eds., vol.~229 of LIPIcs,
  Schloss Dagstuhl - Leibniz-Zentrum f{\"{u}}r Informatik, 2022,
  pp.~60:1--60:17.

\bibitem{FominLSZ19}
{\sc F.~V. Fomin, D.~Lokshtanov, S.~Saurabh, and M.~Zehavi}, {\em
  Kernelization}, Cambridge University Press, Cambridge, 2019.
\newblock Theory of parameterized preprocessing.

\bibitem{GalvezGIHKW21}
{\sc W.~G{\'{a}}lvez, F.~Grandoni, S.~Ingala, S.~Heydrich, A.~Khan, and
  A.~Wiese}, {\em Approximating geometric knapsack via {L}-packings}, {ACM}
  Trans. Algorithms, 17 (2021), pp.~33:1--33:67.

\bibitem{GareyJ79}
{\sc M.~R. Garey and D.~S. Johnson}, {\em Computers and Intractability: {A}
  Guide to the Theory of NP-Completeness}, W. H. Freeman, 1979.

\bibitem{Goldberg70}
{\sc M.~Goldberg}, {\em The packing of equal circles in a square}, Mathematics
  Magazine, 43 (1970), pp.~24--30.

\bibitem{HarrenJPS14}
{\sc R.~Harren, K.~Jansen, L.~Pr{\"{a}}del, and R.~van Stee}, {\em A {(5/3} +
  {\(\epsilon\)})-approximation for strip packing}, Comput. Geom., 47 (2014),
  pp.~248--267.

\bibitem{HochbaumM85}
{\sc D.~S. Hochbaum and W.~Maass}, {\em Approximation schemes for covering and
  packing problems in image processing and {VLSI}}, J. {ACM}, 32 (1985),
  pp.~130--136.

\bibitem{JansenR19}
{\sc K.~Jansen and M.~Rau}, {\em Closing the gap for pseudo-polynomial strip
  packing}, in 27th Annual European Symposium on Algorithms, {ESA} 2019,
  September 9-11, 2019, Munich/Garching, Germany, M.~A. Bender, O.~Svensson,
  and G.~Herman, eds., vol.~144 of LIPIcs, Schloss Dagstuhl - Leibniz-Zentrum
  f{\"{u}}r Informatik, 2019, pp.~62:1--62:14.

\bibitem{J.-Kepler:1611sf}
{\sc J.~Kepler}, {\em Strena seu de nive sexangula}, Frankfurt: Godefrid
  Tampach, 1611.

\bibitem{LiuMS98}
{\sc Y.~Liu, A.~Morgana, and B.~Simeone}, {\em A linear algorithm for 2-bend
  embeddings of planar graphs in the two-dimensional grid}, Discret. Appl.
  Math., 81 (1998), pp.~69--91.

\bibitem{locatelli2002packing}
{\sc M.~Locatelli and U.~Raber}, {\em Packing equal circles in a square: a
  deterministic global optimization approach}, Discrete Applied Mathematics,
  122 (2002), pp.~139--166.

\bibitem{maranas1995new}
{\sc C.~D. Maranas, C.~A. Floudas, and P.~M. Pardalos}, {\em New results in the
  packing of equal circles in a square}, Discrete Mathematics, 142 (1995),
  pp.~287--293.

\bibitem{Mohar01}
{\sc B.~Mohar}, {\em Face covers and the genus problem for apex graphs}, J.
  Comb. Theory, Ser. {B}, 82 (2001), pp.~102--117.

\bibitem{DBLP:conf/focs/NaorSS95}
{\sc M.~Naor, L.~J. Schulman, and A.~Srinivasan}, {\em Splitters and
  near-optimal derandomization}, in 36th Annual Symposium on Foundations of
  Computer Science, Milwaukee, Wisconsin, USA, 23-25 October 1995, 1995,
  pp.~182--191.

\bibitem{nurmela1999more}
{\sc K.~J. Nurmela et~al.}, {\em More optimal packings of equal circles in a
  square}, Discrete \& Computational Geometry, 22 (1999), pp.~439--457.

\bibitem{nurmela1997packing}
{\sc K.~J. Nurmela and P.~R. {\"O}sterg{\aa}rd}, {\em Packing up to 50 equal
  circles in a square}, Discrete \& Computational Geometry, 18 (1997),
  pp.~111--120.

\bibitem{schaer1965densest}
{\sc J.~Schaer}, {\em The densest packing of 9 circles in a square}, Canadian
  Mathematical Bulletin, 8 (1965), pp.~273--277.

\bibitem{specht2015best}
{\sc E.~Specht}, {\em The best known packings of equal circles in a square (up
  to {N= 10000})}, 2015.

\bibitem{szabo2007new}
{\sc P.~G. Szab{\'o}, M.~C. Mark{\'o}t, T.~Csendes, E.~Specht, L.~G. Casado,
  and I.~Garc{\'\i}a}, {\em New approaches to circle packing in a square: with
  program codes}, vol.~6, Springer Science \& Business Media, 2007.

\bibitem{toth2013lagerungen}
{\sc L.~F. T{\'o}th}, {\em Lagerungen in der Ebene auf der Kugel und im Raum},
  Springer, 1953.

\end{thebibliography}

\end{document}